\documentclass[
reprint,
bibnotes,
superscriptaddress,
notitlepage, 
amsfonts,
twocolumn,
maxcitenames=2,
maxbibnames=6,
longbibliography, 
]{revtex4-1}

\usepackage{qcircuit}
\usepackage[dvipdfmx]{graphicx}
\usepackage{amsmath,amssymb,amsthm,mathrsfs,amsfonts,dsfont, mathtools}
\usepackage{subfigure, epsfig}
\usepackage{braket}
\usepackage{bm}
\usepackage{enumerate}
\usepackage{physics}
\usepackage{multirow}
\usepackage{booktabs}

\usepackage{hyperref}
\hypersetup{colorlinks=true,linkcolor=blue,citecolor=blue,urlcolor=blue}

\newtheorem{thm}{Theorem}
\newtheorem{lemma}{Lemma}

\newtheorem{proposition}{Proposition}
\usepackage{algorithm}
\usepackage{algpseudocode}

\newcommand{\mathJ}{\mathcal{J}}

\newcommand{\mathU}{\mathcal{U}}

\algnewcommand\algorithmicinput{\textbf{Input:}}
\algnewcommand\algorithmicoutput{\textbf{Output:}}
\algnewcommand\Input{\item[\algorithmicinput]}%
\algnewcommand\Output{\item[\algorithmicoutput]}%

\newcommand{\black}[1]{\textcolor{black}{#1}}

\newcommand{\ddist}{d_{\diamond}}

\newcommand{\idop}{\mathbb{I}}%identity operator
\newcommand{\rr}{\mathbb{R}}%real number}
\newcommand{\supp}[1]{\text{supp}\left(#1\right)}%support of function

\newcommand{\diamondnorm}[1]{\left\|#1\right\|_{\diamond}}%diamond norm
\newcommand{\lpnorm}[2]{\left\|#2\right\|_{#1}}%Lp norm

\renewcommand{\selectlanguage}[1]{}

\begin{document}

\title{Error Crafting in Mixed Quantum Gate Synthesis}

% ------------  AUTHORS AND AFFILIATIONS ----------

\author{Nobuyuki Yoshioka}
\email{nyoshioka@ap.t.u-tokyo.ac.jp}
\affiliation{Department of Applied Physics, University of Tokyo, 7-3-1 Hongo, Bunkyo-ku, Tokyo 113-8656, Japan}
\affiliation{Theoretical Quantum Physics Laboratory, RIKEN Cluster for Pioneering Research (CPR), Wako-shi, Saitama 351-0198, Japan}
\affiliation{JST, PRESTO, 4-1-8 Honcho, Kawaguchi, Saitama, 332-0012, Japan}
\affiliation{International Center for Elementary Particle Physics, University of Tokyo, 7-3-1 Hongo, Bunkyo-ku, Tokyo 113-0034, Japan}

\author{Seiseki Akibue}
\affiliation{NTT Communication Science Laboratories, NTT Corporation, Atsugi 243-0198, Japan}

\author{Hayata Morisaki}
\affiliation{Graduate School of Engineering Science, Osaka University
1-3 Machikaneyama, Toyonaka, Osaka 560-8531, Japan}

\author{Kento Tsubouchi}
\affiliation{Department of Applied Physics, University of Tokyo, 7-3-1 Hongo, Bunkyo-ku, Tokyo 113-8656, Japan}

\author{Yasunari Suzuki}
\affiliation{NTT Computer and Data Science Laboratories, NTT Corporation, Musashino 180-8585, Japan}

%%%%%%%%%%%%%%%%%%%%%%%%%%%%%%%%%%%%%%%%%%%%%%%%

% we want to emphasize that this is peculiar to FTQC
\begin{abstract} 
In fault-tolerant quantum computing, errors in unitary gate synthesis is comparable with noise inherent in the gates themselves.
While mixed synthesis can suppress such coherent errors quadratically, there is no clear understanding on its remnant error, which hinders us from designing a holistic and practical error countermeasure.
In this work, we propose that the classical characterizability of synthesis error can be exploited; remnant errors can be crafted to satisfy desirable properties.
We prove that we can craft the remnant error of arbitrary single-qubit unitaries to be Pauli and depolarizing errors, while the conventional twirling cannot be applied in general.
For Pauli rotation gates, in particular, the crafting enables us to suppress the remnant error up to {\it cubic} order, which results in synthesis with a T-count of $\log_2(1/\varepsilon)$ up to accuracy of $\varepsilon=10^{-9}$.
Our work opens a novel avenue in quantum circuit design and architecture that orchestrates error countermeasures.
\end{abstract}

\maketitle

\let\oldaddcontentsline\addcontentsline% Store \addcontentsline
\renewcommand{\addcontentsline}[3]{}% Make \addcontentsline a no-op

\section*{Introduction}
It is widely accepted that quantum computers with error correction~\cite{shor1995scheme, gottesman1997stabilizer} can be potentially useful in various tasks such as factoring~\cite{shor1999polynomial}, quantum many-body simulation~\cite{abrams1999quantum}, quantum chemistry~\cite{aspuru2005simulated}, and machine learning~\cite{biamonte2017quantum}.
With the tremendous advancements in quantum technology, the earliest experimental demonstrations of quantum error correction are realized with various codes ~\cite{sundaresan2023demonstrating, egan2021faulttolerant, ryan-anderson2021realizationa, bluvstein2024logical, livingston2022experimental, satzinger2021realizing,bluvstein2022quantum, zhao2022realization, acharya2023suppressing}.
While such progress towards fault-tolerant quantum computing (FTQC) is promising,  quantum computers in the coming generation, the {\it early} fault-tolerant quantum computers, would still be subject to noise at the logical level, and hence the error countermeasures developed for \black{noisy intermediate-scale quantum (NISQ) devices}, \black{in particular the quantum error mitigation (QEM) techniques}~\cite{temme2017error, li2017efficient, endo2018practical, koczor_exponential_2021, huggins_virtual_2021, yoshioka2022generalized, vandenBerg2022model, strikis2021learning}, will remain essential to exemplify practical quantum advantage~\cite{piveteau2021error, suzuki2022quantum, yoshioka2024hunting, babbush_encoding_2018, lee_evenmore_2021, sakamoto2023end, beverland2022assessing}.
Meanwhile, there is a stark difference between (early) FTQC and \black{NISQ} regimes that is characterized by the Eastin-Knill theorem~\cite{eastin2009restrictions}; \black{currently known FTQC schemes are based on universal gate sets that consist of discrete gates.}
%the \black{universal} gate set is discrete.
This indicates that the gate synthesis yields coherent error whose effect is comparable to incoherent errors triggered by noise in the hardware.

In contrast to incoherent errors which inevitably necessitate entropy reduction via measurement and feedback~\cite{aharonov1996limitations, tsubouchiUniversalCostBound2023, takagi2023universal, quek2022exponentially}, there is no corresponding fundamental limitation for coherent errors.
In fact, several previous works have found that algorithmic errors can be dealt with more efficiently than naive expectation~\cite{hastings_turning_2017, campbell_shorter_2017, wallman2016noise, Kliuchnikov2023shorterquantum, akibue2023probabilistic, akibue2024probabilistic}. 
It was pointed out by Refs.~\cite{hastings_turning_2017, campbell_shorter_2017} that, \black{when one has access to a Pauli rotation exposed to over and under rotations, we can simply take a classical mixture of two to achieve quadratic suppression of the error.} 
Such findings have invoked the proposal of a mixed synthesis scheme that approximately halves the consumption of magic gates under single-qubit unitary synthesis~\cite{Kliuchnikov2023shorterquantum, akibue2024probabilistic}.
It must be noted that, these existing works merely focus on the accuracy of the mixed synthesis, while one must take care of the compatibility between other error countermeasures to further suppress the errors. \black{This is prominent when one desires to \black{obtain unbiased estimator} via the QEM, which requires a deep knowledge of the error channel.}
To our knowledge, however, there is no firm understanding on the remnant error channel nor attempt to control the profile of them.

In this work, we make a first step to fill this gap by introducing a notion of {\it error crafting}, i.e., perform the mixed synthesis that results in a remnant error channel with a desired property, for single-qubit unitaries synthesized under the framework of Clifford+T formalism.
We rigorously prove that one can design a synthesis protocol that ensures the remnant error is constrained to Pauli error, which is crucial to perform efficient error mitigation techniques~\cite{temme2017error, li2017efficient, suzuki2022quantum}.
\black{We further provide numerical evidence that, in the case of Pauli rotation gates, one can reliably manipulate the remnant error to be biased. This allows us to perform error detection at the logical level which achieves quadratically small sampling overhead compared to the existing lower bound of sample complexity that does not rely on mid-circuit measurement~\cite{tsubouchiUniversalCostBound2023}. We furthermore extend our framework to perform error crafting using CPTP channels obtained from error correction at the logical layer, and show that the remnant error can be suppressed up to cubic order.
As a result, we can implement Pauli rotation gates of accuracy $\varepsilon$ with T-count of $\log_2(1/\varepsilon)$ up to $\varepsilon=10^{-9}$, which is expected to be sufficient for various medium-scale FTQC beyond state-of-the-art classical computation~\cite{babbush_encoding_2018, lee_evenmore_2021, yoshioka2024hunting, beverland2022assessing}.}

We remark that the operation of twirling shares the spirit, since it also transforms a noise channel into another channel with better-known properties via randomized gate compilation.
\black{However, there are two severe limitations in twirling: first, it agnostically smears out the characteristics of the channel, and second, \black{it cannot ensure the noise maps to be stochastic when one considers non-Clifford operations beyond the third level of Clifford hierarchy (see Table~\ref{tab:applicability} and Supplementary Note~\ref{sec:twirling-limitation} for details).  
We note that existing works on randomized compiling report suppression of the noise~\cite{santos2024pseudo, odake2024robust, wallman2016noise}, while we are not aware of any work that guarantees the structure of the noise channel.}
This is in sharp contrast to our work, since our objective is to fully utilize the information on the error profile of general unitary synthesis that can be characterized completely and efficiently via classical computation.}

%\black{mention table 1}

\begin{table}[t]
    \centering
    \begin{tabular}{c|c|c|c}
         & Pauli twirl & Clifford twirl & Crafting  \\
         \hline
         T gate& $\checkmark$ & Partially & $\checkmark$ \\
         Z rotation & Partially & Partially & $\checkmark$\\
         Arbitrary & $\times$ & $\times$ & $\checkmark$
    \end{tabular}
    \caption{\black{Applicability of noise shaping techniques for unitary synthesis error in single-qubit non-Clifford gates. 
    The row ``T gate" represents the applicability for all single-qubit gates from the third level of the Clifford hierarchy. Of course, the synthesis error is not present if one includes the T gate in the implementable operations, since there is no need for crafting.
    The row ``Z rotation" represents the applicability for Pauli rotation gates. Here, ``Partially" means that we can employ a symmetric Clifford group that commutes with the target unitary to twirl the noise within each irreducible representation but not the entire space~\cite{mitsuhashi2023clifford, tsubouchi2024symmetric}. For instance, over- and under-rotation errors cannot be twirled; one must rely on crafting.}}
    \label{tab:applicability}
\end{table}

\section*{Results}
\subsection*{Unitary, mixed, and error-crafted synthesis}
We aim to implement a target single-qubit unitary $\hat{U}$ with its channel representation given by $\hat{\mathcal{U}}$. We hereafter distinguish the unitaries before and after synthesis by the presence and absence of a hat.
Due to the celebrated Eastin-Knill's theorem~\cite{eastin2009restrictions}, there is no fault-tolerant implementation of continuous gate sets, and therefore we cannot exactly implement the target in general. 
\black{We alternatively aim to approximate $\hat{U}$ by decomposing it to a sequence of implementable operations.
Such a task is called {\it gate synthesis}.}

\black{One of the most common approaches for gate synthesis is the unitary synthesis, namely to decompose the target unitary into a sequence of implementable unitary gate sets.  
For instance, unitary synthesis into the Clifford+T gate set is known to be possible with classical computation time of $O({\rm polylog}(1/\epsilon))$ for arbitrary single-qubit Pauli rotation assuming a factoring oracle~\cite{ross2016optimal} and $O({\rm poly}(1/\epsilon))$ for a general SU(2) gate~\cite{fowler2011constructing, bocharov2012resource}.
\black{Another measure for the efficiency of synthesis is} the number of T gates, or the T-count, since T gates require cumbersome procedures such as magic state distillation and gate teleportation. 
\black{In this regard,} the existing methods \black{nearly} saturates the information-theoretic lower bound regarding the T-count of $3 \log_2(1/\epsilon)$ except for rare exceptions that require $4 \log_2(1/\epsilon)$\black{~\cite{kliuchnikov2013fast, ross2016optimal, fowler2011constructing, bocharov2012resource}}.
}

\black{
%If one is allowed to repeat the circuit execution as in early FTQC scheme, 
We may also consider mixed synthesis, which takes a classical mixture of several different outputs from the unitary synthesis algorithm. 
Concretely, we generate a set of slightly shifted unitaries $\{U_j\}$ which are $\epsilon$-close approximations of the target unitary $\hat{U}$ in terms of diamond distance as $\ddist(\hat{\mathcal{U}}, \mathcal{U}_j) \leq \epsilon$ (see Supplementary Note~\ref{sec:formalism_prob_synth} for the definition of diamond distance).}
It is known that the mixed synthesis allows quadratic suppression of the error, when we appropriately generate a set of synthesized  unitaries $\{\mathU_j\}$ and choose the probability distribution over them:

\begin{lemma}(informal~\cite{akibue2024probabilistic})
Let $\hat{\mathcal{U}}$ be a target unitary channel with $\{\mathcal{U}_j\}$ constituting its $\epsilon$-net in terms of the diamond distance. Then, one can find the optimal weight $\{p_j\}$ such that the diamond distance is bounded as
\begin{eqnarray}\label{eq:quad_suppression}
\ddist\bigg(\hat{\mathcal{U}}, \sum_j p_j \mathcal{U}_j \bigg) \leq \epsilon^2,
\end{eqnarray}
by solving the following minimization problem via semi-definite programming:
\begin{eqnarray}
    p = \underset{p} {\operatorname{argmin}} ~\ddist
    \bigg(\hat{\mathcal{U}}, \sum_j p_j \mathcal{U}_j \bigg).\label{eq:minimization}
\end{eqnarray}
\end{lemma}

\black{As an extension of existing synthesis techniques, here we propose {\it error-crafted synthesis}, or {\it error crafting}. In short, this can be understood as constraining the remnant error of the synthesis, depending on the subsequent error countermeasure (such as QEM or error detection) one wishes to perform. We formally define the remnant error channel as
\begin{eqnarray}
    \mathcal{E}_{\rm rem}(\cdot) = \left(\sum_j p_j \Lambda_j\right) \circ \hat{\mathcal{U}}^{-1}(\cdot) = \sum_{ab} \chi_{ab} P_a \cdot P_b,
\end{eqnarray}
where $P_a \in \{{\rm I}, {\rm X}, {\rm Y}, {\rm Z}\}$ and $\{\chi_{ab}\}$ gives the chi representation of the remnant error. We also allow the implementable channel to be a non-unitary CPTP map as $\{\Lambda_j\},$ which we find to be useful when we incorporate measurement and feedback as discussed later.
The aim of the crafted synthesis is to control $\{\chi_{ab}\}$ by optimizing the probability distribution $\{p_j\}$ for appropriately generated $\{\Lambda_j\}$.}

\black{
Let us remark that we require the synthesis protocol to be distinct from the QEM stage, in a sense that we do not allow synthesis to increase the sampling overhead, the multiplication factor in the number of circuit execution.
Also, we require that the gate complexity is nearly homogeneous even if we include feedback operations, in order to prevent a situation where parallel application of multiple synthesized channels is stalled due to a fallback scheme~\cite{bocharov2015efficient, Kliuchnikov2023shorterquantum}. 
}

\begin{figure}[tbp]
    \centering    \includegraphics[width=0.85\linewidth]{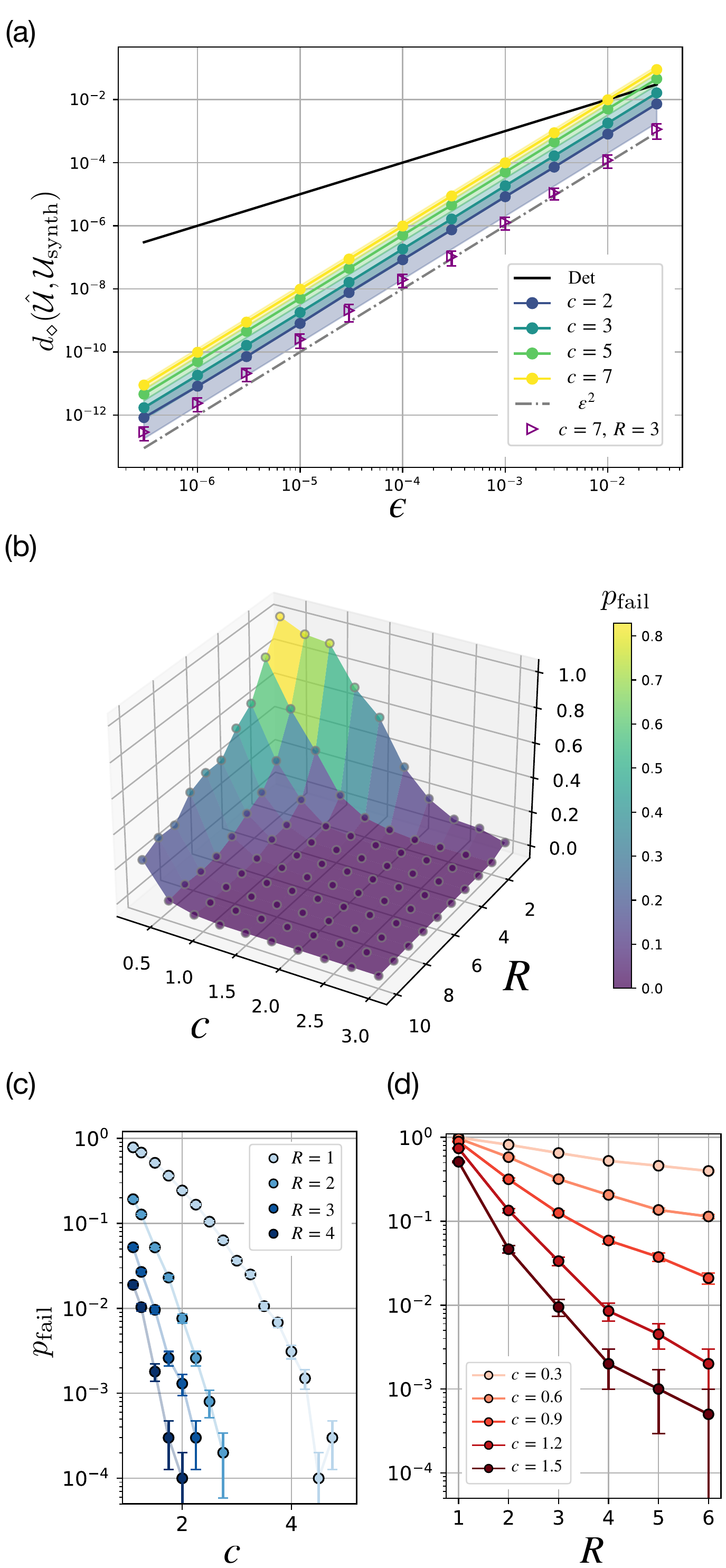}
    \caption{{\bf Numerical analysis on error-crafted synthesis for the Pauli constraint.}
    (a) Diamond distance $\ddist(\hat{\mathU}, \mathU_{\rm synth})$ between the target single-qubit Haar random unitary $\hat{\mathU}$ and the synthesized channel $\mathU_{\rm synth} \coloneqq\sum_j p_j \mathU_j$. The black real line indicates the results from unitary synthesis by the Ross-Selinger algorithm, while the filled dots are results from mixed synthesis with shift factors of $c=2, 3, 5, 7.$ When we increase the number of synthesized unitaries by a factor of $R=3$~\black{(see main text for detailed definition)} for $c=7$, shown by unfilled triangles, the remnant synthesis error is suppressed below the bound in Eq.~\eqref{eq:proberrorPauli_main} and approaches the value of $\epsilon^2$, which matches the upper bound for the non-constrained case as in Eq.~\eqref{eq:quad_suppression}. 
    The plots are averaged over 200 \black{random instances of target unitaries}.
    (b) Surface plot of failure rate $p_{\rm fail}$ of mixed synthesis under the Pauli constraint, averaged over 200 \black{instances} for $\epsilon=10^{-4}.$ 
    Scaling with $c$ and $R$ are shown in (c), (d), where the number of \black{instances} is 10000 and 2000, respectively.
    }
    \label{fig:dnorm_quadratic}
\end{figure}

\subsection*{Crafting remnant error to be Pauli channel}
While the quadratic suppression~\eqref{eq:quad_suppression} by mixed synthesis is already beneficial for practical quantum computing, the performance of many error counteracting protocols can be enhanced when we craft the remnant error $\mathcal{E}_{\rm rem}$.
For instance, let us consider the probabilistic error cancellation (PEC), which is one of the most well-known error mitigation techniques that estimate expectation values of physical observables by the quasi-probabilistic implementation of error channel inversion~\cite{temme2017error, vandenberg2023probabilistic}.
One can show that it suffices to take Clifford operations and Pauli initialization to constitute a universal set to invert arbitrary single-qubit error channels~\cite{endo2018practical, takagi2021optimal}, while the implementation is significantly simplified when $\mathcal{E}_{\rm rem}$ is a Pauli channel.
This is because the PEC can be performed merely via updates on Pauli frames, which is purely done by classical processing of stabilizer measurements~\cite{suzuki2022quantum}.
As another example, when we consider the rescaling technique under white-noise approximation~\cite{arute2019quantum, dalzell2024random}, which is a provably cost-optimal way of error mitigation, \black{we can drastically improve the accuracy of the approximation when} the errors are free from coherent components (see Supplementary Note~\ref{sec:wn_coherent} for details).

In this regard, we make a further step and ask the following question: {\it can we guarantee both quadratic error suppression and success of crafting the remnant error to be Pauli channel?}
\black{This is equivalent to solving the following optimization problem:
\begin{align}
\label{eq:Paulierror_main}
{\rm minimize}~ 
d_{\diamond}\left(\hat{\mathcal{U}}, \sum_j p_j \mathcal{U}_j\right)~
{\rm s.t.}~\mathcal{E}_{\rm rem}(\cdot)=\sum_{a}\chi_{aa} P_a\cdot P_a, 
\end{align}
where $\chi_{aa}\geq0$ is required for all $a$.}
\black{We emphasize that this cannot be achieved by naively applying twirling techniques, since randomization for coherent errors following non-Clifford unitary demands non-Clifford operations, which are not error-free in FTQC \black{(see Supplementary Note~\ref{sec:twirling-limitation})}.}
Our strategy is to prepare shift unitaries $\{\mathcal{V}_j^{c\epsilon}\}_j$ that can be used to generate a set of $c\epsilon$-close unitaries such that the feasibility of the minimization problem~\eqref{eq:Paulierror_main} is guaranteed, even when unitary synthesis induces error of $\epsilon$.
By explicitly constructing such a set of shift unitaries, we rigorously answer to the above-raised question as the following:
\begin{thm}\label{thm1_informal}
There exist shift factor $c$, a positive number $\epsilon_0$, and a set of shift unitaries $\{\mathcal{V}_j^{c\epsilon}\}_{j=1}^7$
 such that,
for any given single-qubit unitary $\hat{\mathcal{U}}$ and for any $\epsilon\in(0,\epsilon_0]$, if we synthesize the shifted target unitaries to obtain $\{\mathcal{U}_j\} = \{\mathcal{A}(\mathcal{V}_j^{c\epsilon} \circ \hat{\mathU})\} $ so that \black{$\ddist(\mathcal{V}_j^{c\epsilon}\circ \hat{\mathcal{U}}, \mathcal{U}_j)\leq \epsilon$ }is satisfied under the synthesis algorithm $\mathcal{A}$,
then there exists a mixture of $\{\mathcal{U}_j\}$ that induces a Pauli channel as the effective synthesis error, with the accuracy being
\begin{equation}
\label{eq:proberrorPauli_main}
 (c-1)^2\epsilon^2 \leq 
 \ddist\bigg(\hat{\mathcal{U}}, \sum_j p_j \mathcal{U}_j \bigg)
 \leq (c+1)^2 \epsilon^2.
\end{equation}
\end{thm}

\black{Remarkably, this theorem states that} we can craft the error without assuming anything but the accuracy for the synthesis algorithm $\mathcal{A}$.
It is also worth noting that the minimization problem~\eqref{eq:Paulierror_main} can be solved by the linear programming.

To prove Theorem~\ref{thm1_informal}, we explicitly construct the shift unitaries $\{\mathcal{V}_{j}^{c\epsilon}\}_j$, and show that the feasibility of the solution is robust even under additional synthesis error of $\epsilon.$ Once such a solution is ensured, the upper bound of Eq.~\eqref{eq:proberrorPauli_main} follows directly from the formulation via the minimization over diamond distance and the fact that $\{\mathcal{U}_j\}_{j=1}^7$ is $(c+1)\epsilon$-close to $\hat{\mathU}$ with respect to the diamond distance. Also, the lower bound follows from the fact that the diamond distance between the compiled shifted unitaries $\{\mathcal{U}_j\}$ and the target unitary $\hat{\mathU}$ are at least $(c-1)\epsilon$. 
We guide readers to Supplementary Note~\ref{sec:guarantee_crafting} for details of the proof.

Since the theorem only states the existence of finite $c$, we rely on numerical simulation to verify the practical usage with finite $c$ values.
For this purpose, we synthesize a single-qubit Haar random unitary gate by first performing the axial decomposition into three Pauli rotations, each of which passed to the synthesis algorithm proposed by Ross and Selinger~\cite{ross2016optimal}. 
As shown in Fig.~\ref{fig:dnorm_quadratic}(a), the mixed synthesis indeed yields error-bounded solutions that achieve quadratic suppression compared to the unitary synthesis.
Meanwhile, we find that the Pauli constraint is violated when we take small shift factors $c$. This is not contradicting Theorem~\ref{thm1_informal}, since the theorem only states the existence of $c$ which always yield successful results. 
It must be noted that, as shown in Fig.~\ref{fig:dnorm_quadratic}(b) and (c), the occurrence of such a violation is suppressed exponentially by increasing the value of $c$. We observe that it suffices to take $c\sim3.5, 4.5$ to yield feasible solution with failure probability of $10^{-2}, 10^{-3}$, respectively, and naturally deduce that $c\sim 5.5, 6.5$ is required to reach $10^{-4}, 10^{-5}.$
We remark that, while we have employed the axial decomposition of the SU(2) gates for the results shown here, we observe similar suppression when we utilize a direct synthesis algorithm~\cite{morisaki2024}.

\black{Since increasing the shift factor $c$ increases the prefactor of the remnant error by $c^2$, we find it beneficial to introduce another practical strategy that enhances the feasibility of the solution. }
Here we provide a simple yet powerful technique; by generating the shift unitaries for $R$ different shift factors $c/R, 2c/R, ..., c$ and solving the problem~\eqref{eq:Paulierror_main} with $R$-fold larger number of basis set, we can both suppress the failure of error-crafted synthesis (see Fig.~\ref{fig:dnorm_quadratic}(b) and (d)) and the remnant error.
As shown in Fig.~\ref{fig:dnorm_quadratic}(a), when we take $c=7$ and $R=3$,  we find that the diamond distance is suppressed as $\epsilon^2$ for a broad range of the synthesis accuracy $\epsilon$, which is a significant improvement from the case of $R=1$ with $\sim c^2 \epsilon^2$. 
Note that this is extremely beneficial when we wish to eliminate such remnant errors using QEM methods. 
To be concrete, the sampling overhead is given as $(1+4p)^L$ under error rate of $p$ with $L$ gates when one employs the PEC method, and hence the reduction of the remnant error from $p_{\rm rem}\sim c^2\epsilon^2$ to $ p_{\rm rem}\sim \epsilon^2$ results in exponential difference with $L$.
For instance, when we implement $\{1, 2, 3\}\times10^8$ rotation Pauli gates with  $\epsilon=10^{-5}$, the sampling overhead is reduced by factors of 6.8, 47, and 317, respectively.

\begin{figure}[bh]
    \centering    \includegraphics[width=0.85\linewidth]{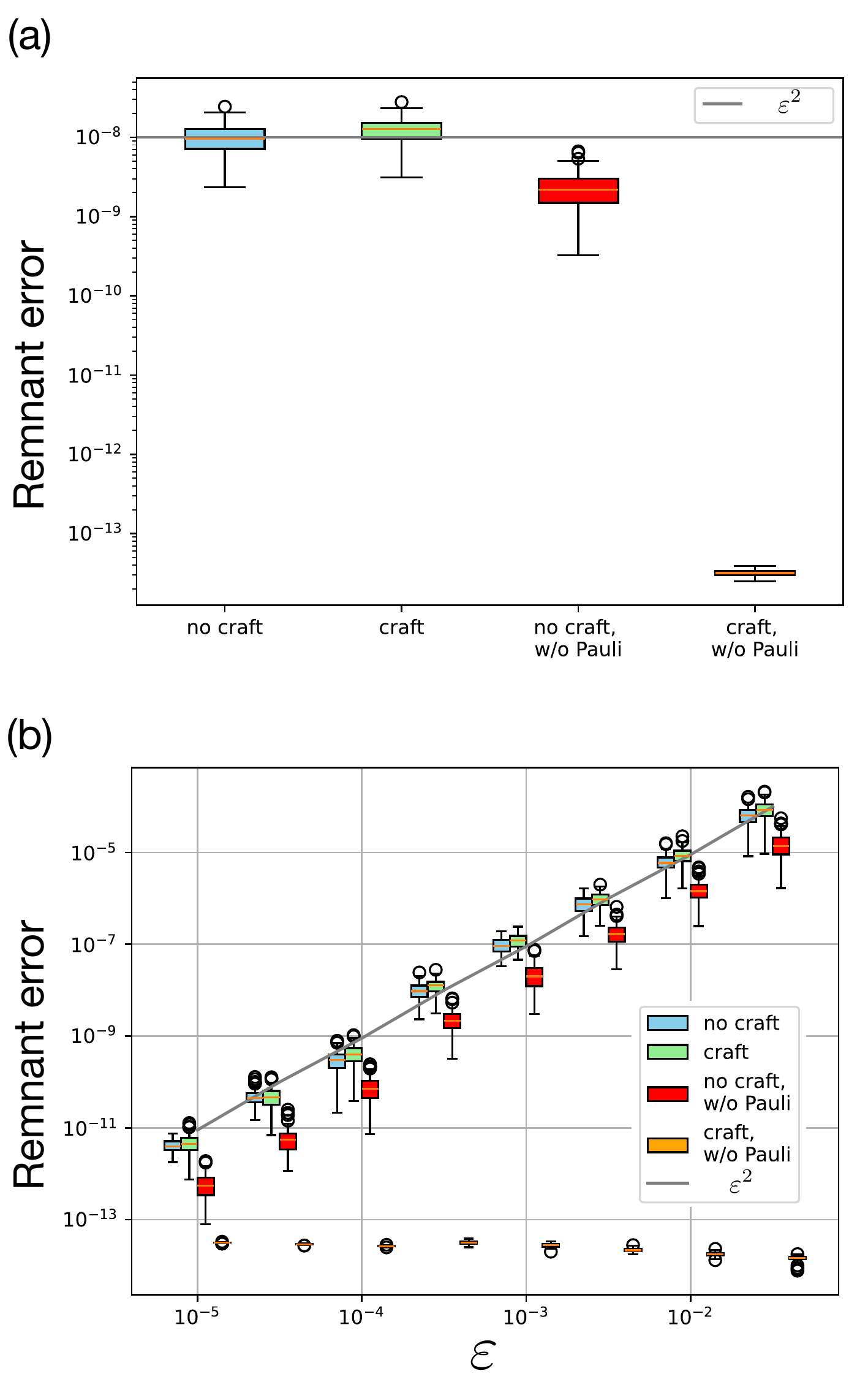}
    \caption{\black{{\bf Comparison of the magnitude of remnant errors $d_{\diamond}(\mathcal{E_{\rm rem}}, \mathcal{I})$.} (a) Box plots are provided for the magnitude of remnant errors of 100 random instances of Z rotation gate, with unitary synthesis done for $\epsilon=10^{-4}$. The left two box plots show $d_{\diamond}(\mathcal{E}_{\rm rem}, \mathcal{I})$ with and without crafting for Pauli error, and the right two box plots indicate the magnitude of the non-Pauli contribution. (b) Scaling of the remnant errors with $\epsilon.$
    In both panels, shift unitaries are generated with $c=3, R=5.$
    }
    }
    \label{fig:crafted-vs-noncrafted-main}
\end{figure}

\black{We remark that crafting is actually essential to guarantee the remnant channel to be a Pauli channel. In Fig.~\ref{fig:crafted-vs-noncrafted-main}, we compare the remnant errors in crafted and uncrafted mixed synthesis.
We find that, as expected, the magnitude of the remnant error is slightly smaller when we do not craft it. However, we can see there is a drastic difference in the non-Pauli component. Here, we calculate the magnitude of the error that persists even after we eliminate the Pauli component as $\min_{\mathcal{R}} d_{\diamond}(\mathcal{R}\circ \mathcal{E}_{\rm rem}, \mathcal{I})$, where $\mathcal{R}$ is minimized over the set of trace-preserving maps that inverts Pauli error. We can see that the non-Pauli component in uncrafted synthesis remains as $O(\epsilon^2)$, while in the crafted synthesis the magnitude is suppressed below $10^{-13}$. We attribute this residual to the relaxation of equality constraint into inequality, which was performed to enhance the stability of the numerical solver (see Supplementary Note~\ref{sec:numerics_detail}). While we show the result for Z rotation here, we have confirmed that similar results can be obtained for general SU(2) gates as well.
}

\black{As another practical remark,} note that the number of bases with nonzero probability ($\sum_{p_j\neq 0} 1$) is always upper-bounded by 10. Hence, the actual memory footprint of the transpiled circuit is comparable to ordinary mixed synthesis algorithms.

\subsection*{Crafting remnant error to be detectable}
If one chooses to employ the PEC method as the QEM method for FTQC, it suffices to guarantee only the feasibility under the Pauli constraint.
\black{Indeed, the PEC method is powerful in the sense that one may perform unbiased estimation when the error channel is precisely known.
However, we find that the remnant errors can be counteracted much more effectively by utilizing the fact that they can be characterized completely by classical computation.
For instance, we may wish to craft local depolarizing errors, since
many limits and performance of quantum algorithms, as well as QEM methods, are analyzed under the assumption of the local depolarizing noise~\cite{aharonov1996limitations, koczor_exponential_2021, koczor2021dominant, quek2022exponentially, cai2023quantum}.
We may alternatively wish to craft the remnant errors to be some biased Pauli noise. 
We find that this is extremely beneficial for the case of the Pauli rotation gates. Namely, we can perform error detection or correction in the {\it logical} layer, which reduces the sampling overhead to the extent that QEM methods designed for general channels cannot achieve.
}

First, let us discuss the crafting of the depolarizing channel.
Following a similar strategy as in Theorem~\ref{thm1_informal}, we can mathematically prove that there exists a set of shift unitaries $\{\mathcal{V}^{c\epsilon}_{j}\}$ such that we can achieve quadratic suppression of error under the depolarizing constraint, with the sole difference being that the number of shift unitaries is 9 instead of 7. The formal statement and the proof are provided in Supplementary Note~\ref{sec:guarantee_crafting}. Our numerical simulation also shows similar favorable properties as seen in the case of Pauli constraint (See Supplementary Note~\ref{sec:numerics_detail}). 

Next, let us consider crafting the remnant error to be biased. 
In particular, we focus on the Pauli Z rotation gate and aim to bias the error to be XY error, since it can be detected using a single ancilla qubit with less sampling overhead than the lower bound for QEM methods that do not rely on mid-circuit measurement~\cite{tsubouchiUniversalCostBound2023} (concrete circuit structure provided in Fig.~\ref{fig:detection_circuit}).
While we have not found a guaranteed set of shift unitaries, \black{we find that the nine shift unitaries for the depolarizing noise are of practical use, if we allow violation of the constraint  as $p_z \leq 0.01 p$ where $p\coloneqq p_x+p_y+p_z$
 (see Supplementary Notes~\ref{sec:numerics_detail} for details).}%employ a general strategy of expanding the size of basis set until the synthesis is successful, i.e., the contribution of Z error is below a fixed threshold value~(see Supplementary Note~\ref{sec:numerics_detail} and~\ref{sec:craft-CPTP}).}

\black{The success of error crafting is highlighted in Fig.~\ref{fig:error_distribution}.
As can be seen from Fig.~\ref{fig:error_distribution}(a), while the rate of Pauli errors $(\frac{p_x}{p},\frac{p_y}{p},\frac{p_z}{p})$ is distributed quite homogeneously if we do not impose any constraint on the error profile, we find that the constraints can be reflected with high accuracy. 
We verify that $p_x=p_y=p_z$ is satisfied up to the order of machine precision for the depolarizing constraint,
while for the XY-biased noise, we find that we cannot completely eliminate the Z error; there is a small residual component of $p_z=O(\epsilon^3)$ \black{even if we allow the XY-biased noise to consist of non-unital components. We also attempted to craft X-biased noise, while it did not succeed in most instances.}
}

\begin{figure}[t]
    \centering
    \includegraphics[width=0.98\linewidth]{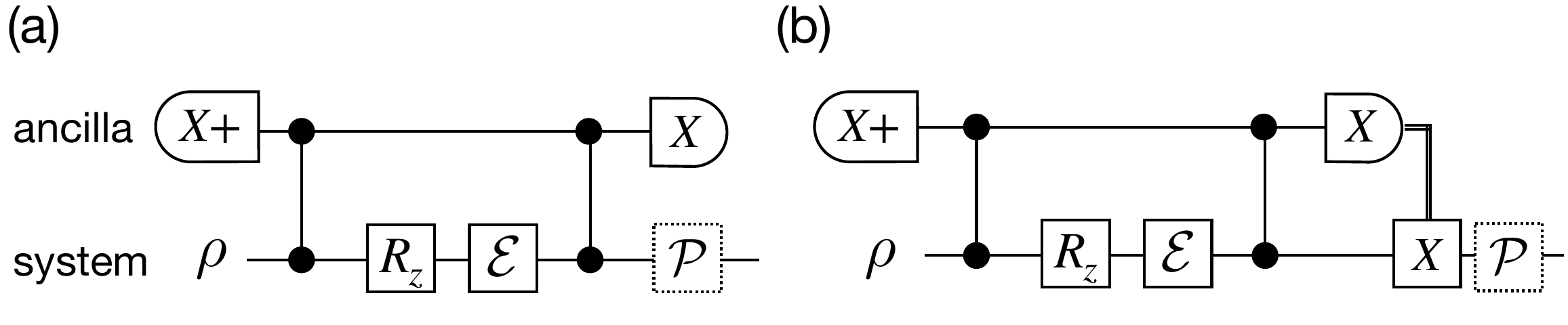}
    \caption{Circuit structure for (a) error detection and (b) logical error correction for Pauli Z rotation gate $R_z$ with remnant error of $\mathcal{E}_{\rm rem}$. We may further combine with the PEC method, for which the Pauli operation $\mathcal{P}$ with dotted square is executed virtually via updating the Pauli frame.
    Note that the feedback operation in (b) is also done via the Pauli frame.
    }
    \label{fig:detection_circuit}
\end{figure}

\begin{figure}[tbp]
    \centering
    \includegraphics[width=0.85\linewidth]{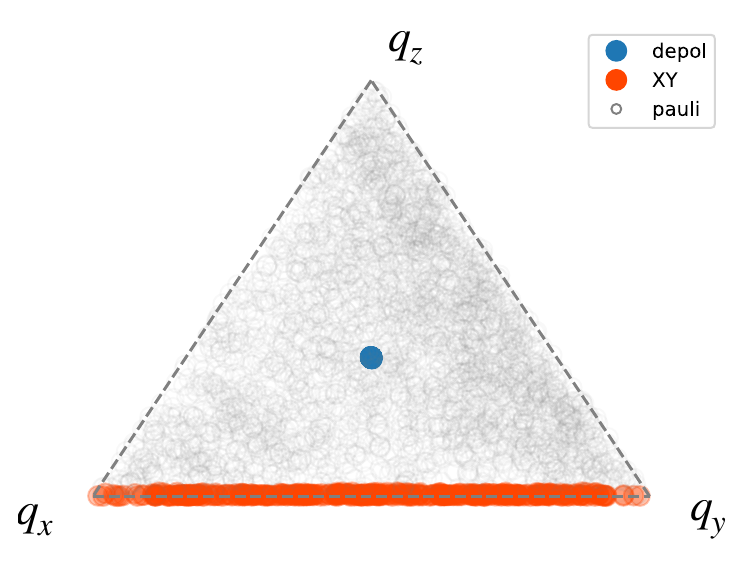}
    \caption{{\bf Numerical analysis on error crafting for remnant errors.} Here we display the distribution of error rate ratio $(q_x, q_y, q_z) \coloneqq (\frac{p_x}{p}, \frac{p_y}{p}, \frac{p_z}{p})$ \black{on a plane $q_x+q_y+q_z=1$, where $p_a \coloneqq \chi_{aa}~(a\in\{x, y, z\})$ is the rate of each Pauli channel in $\mathcal{E}_{\rm rem}$ and $p\coloneqq p_x+p_y+p_z=O(\epsilon^2)$ corresponds to the total error rate. 
    Note that here we show results whose violation of the constraints are below threshold values (see Supplementary Notes~\ref{sec:numerics_detail} for details).
    %We show successful results of crafted synthesis for  depolarizing and XY constraints; the former yields $q_x= q_y= q_z =1/3$ and the latter gives $q_x+q_y= 1-O(\epsilon).$
    } 
    The unitary synthesis is done for $\epsilon=10^{-4}$ with shift factor $c=2.0, 5.0$ for depolarizing and XY constraints, respectively. 
    }
    \label{fig:error_distribution}
\end{figure}

\subsection*{Error crafting using CPTP maps}
\black{While we have considered the error crafting using synthesized unitaries $\{U_j\}$, we can also apply the framework to CPTP maps.
As a demonstration, here we consider the target unitary $\hat{U}$ to be a Pauli Z rotation, and constitute the set of implementable CPTP maps $\{\Lambda_j\}$ as the channels of synthesized unitaries accompanied with entangled measurement and feedback as shown in Fig.~\ref{fig:detection_circuit}(b).
%By performing brute-force search of $\{\Lambda_j\},$ we 
\black{We have numerically found}
that, the remnant error can be suppressed {\it cubically} small as $\epsilon^3$, as opposed to the quadratic suppression realized in the error crafting for unitaries. 
\black{Since the optimal unitary synthesis yield T-count of $3\log_2(1/\epsilon) +O(1)$,}
 such an improvement allows us to synthesize a Pauli rotation gate with a T-count of $\log_2(1/\varepsilon)+O(1)$ on average. \black{Compare this with the scaling of $1.5 \log_2(1/\varepsilon)$ realized by mixed synthesis, which achieve quadratic improvement in $\varepsilon$ (see  Fig.~5).}
We guide the readers to the Supplementary Note~\ref{sec:craft-CPTP},~\ref{sec:shift-unitary-gen} for a detailed discussion on the search for appropriate $\Lambda_j$.
}

\begin{table*}[t]
\begin{tabular}{c||c|c|cc||c|c||c}
\toprule
\multirow{2}{*}{Gate} & \multirow{2}{*}{Synthesis} & \multirow{2}{*}{Crafted?} & \multicolumn{2}{c||}{Remnant error}                & \multirow{2}{*}{QEM method} & \multirow{2}{*}{Sampling overhead}  & \multirow{2}{*}{Comment} \\ \cline{4-5}
                      &           &                 & \multicolumn{1}{c|}{magnitude}    & type          &                      &                                                              \\ \hline
SU(2)                 & Unitary    &                  & \multicolumn{1}{c|}{$\epsilon$}   & coherent       & PEC                  & $1+4\epsilon$                      
& \\
SU(2)                 & Unitary      &               & \multicolumn{1}{c|}{$\epsilon$}   & coherent       & Rescaling            & $1+2\epsilon^2$                      & Slow WN convergence                             \\ \hline
%SU(2)                 & Probabilistic                    & \multicolumn{1}{c|}{$c^2\epsilon^2$} & unital & (PEC)                  & 1 + $O(\epsilon^2)$&           \\
%SU(2)                 & Crafted                    & \multicolumn{1}{c|}{$c^2\epsilon^2$} & Pauli/depol & Rescaling            & $1+2 c^2\epsilon^2$                   & \checkmark \\
SU(2)                 & \black{Mixed Unitary} &                     & \multicolumn{1}{c|}{$\epsilon^2$} & \black{not Pauli} & PEC                  & $1+4 \epsilon^2$    & \black{Pauli frame unavailable}                \\
SU(2)                 & \black{Mixed Unitary} & \checkmark                    & \multicolumn{1}{c|}{$\epsilon^2$} & Pauli/depol & PEC                  & $1+4 \epsilon^2$                    \\
SU(2)                 & \black{Mixed Unitary}& \checkmark                     & \multicolumn{1}{c|}{$\epsilon^2$} & Pauli/depol & Rescaling            & $1+2 \epsilon^2$                                      \\ \hline \hline
$R_z$                 & Unitary &                    & \multicolumn{1}{c|}{$\epsilon$}   & coherent       & PEC                  & $1+4\epsilon$    & \black{Pauli frame unavailable}                   \\
$R_z$                 & Unitary   &                  & \multicolumn{1}{c|}{$\epsilon$}   & coherent       & Rescaling            & $1+2\epsilon^2$                      & Slow WN convergence                              \\ \hline
$R_z$                 & \black{Mixed Unitary}   &  & \multicolumn{1}{c|}{$\epsilon^2$} & \black{not Pauli} & PEC             & $1+4\epsilon^2$     & \black{Pauli frame unavailable}\\
$R_z$                 & \black{Mixed Unitary}   & \checkmark           & \multicolumn{1}{c|}{$\epsilon^2$} & Pauli/depol/XY & PEC             & $1+4\epsilon^2$     \\
$R_z$                 & \black{Mixed Unitary} & \checkmark              & \multicolumn{1}{c|}{$\epsilon^2$} & Pauli/depol/XY & ED + PEC             & $1+(q_x+q_y + 4q_z)\epsilon^2$      & $q_z\sim \epsilon$ in numerics \\
$R_z$                 & \black{Mixed Unitary} & \checkmark              & \multicolumn{1}{c|}{$\epsilon^2$} & Pauli/depol/XY & ED + Resc.           & $1+(q_x+q_y + 2q_z)\epsilon^2$      & $q_z\sim \epsilon$ in numerics\\ \hline
$R_z$                 & \black{Mixed CPTP}& \checkmark                    & \multicolumn{1}{c|}{$\epsilon^3$} & Pauli Z       & PEC                  & $1+ 4\epsilon^3$  & \black{Single ancilla}                  \\
$R_z$                 & \black{Mixed CPTP}& \checkmark                    & \multicolumn{1}{c|}{$\epsilon^3$} & Pauli Z       & Rescaling                & $1+ 2\epsilon^3$   & \black{Single ancilla}                \\
\bottomrule
\end{tabular}
\caption{{\bf Sampling overhead of error countermeasures under synthesis error of $\epsilon$.} The remnant error indicates the diamond distance $d(\mathcal{E}_{\rm rem}, \mathcal{I})$ where $\mathcal{I}$ is the identity channel. We assume that unitary gate synthesis is done with an accuracy of $\epsilon$, and compare the performance when we further perform mixed or error-crafted synthesis. 
The sampling overhead $\gamma$ is indicated for a single gate, and the total overhead is multiplicative as $\prod_l \gamma_l$ when overhead is $\gamma_l$ for the $l$-th gate. The rescaling technique assumes that the underlying quantum circuit is sufficiently deep so that the white noise (WN) approximation holds~\cite{dalzell2024random, morvan2024phase, tsubouchi2024symmetric}.
We assume that the white noise approximation holds when the rescaling technique is applied.
The cubic suppression in the last two rows is verified numerically up to $\epsilon^3=10^{-9}$ as shown in Fig.~\ref{fig:accuracy_tradeoff}.
\black{Note that applying PEC for crafted and uncrafted synthesis differs significantly. We cannot rely on the Pauli frame in the latter, and an actual single-qubit operation is required.}
}
\label{tab:error}
\end{table*}

\begin{figure}[tbp]
    \centering
    \includegraphics[width=0.98\linewidth]{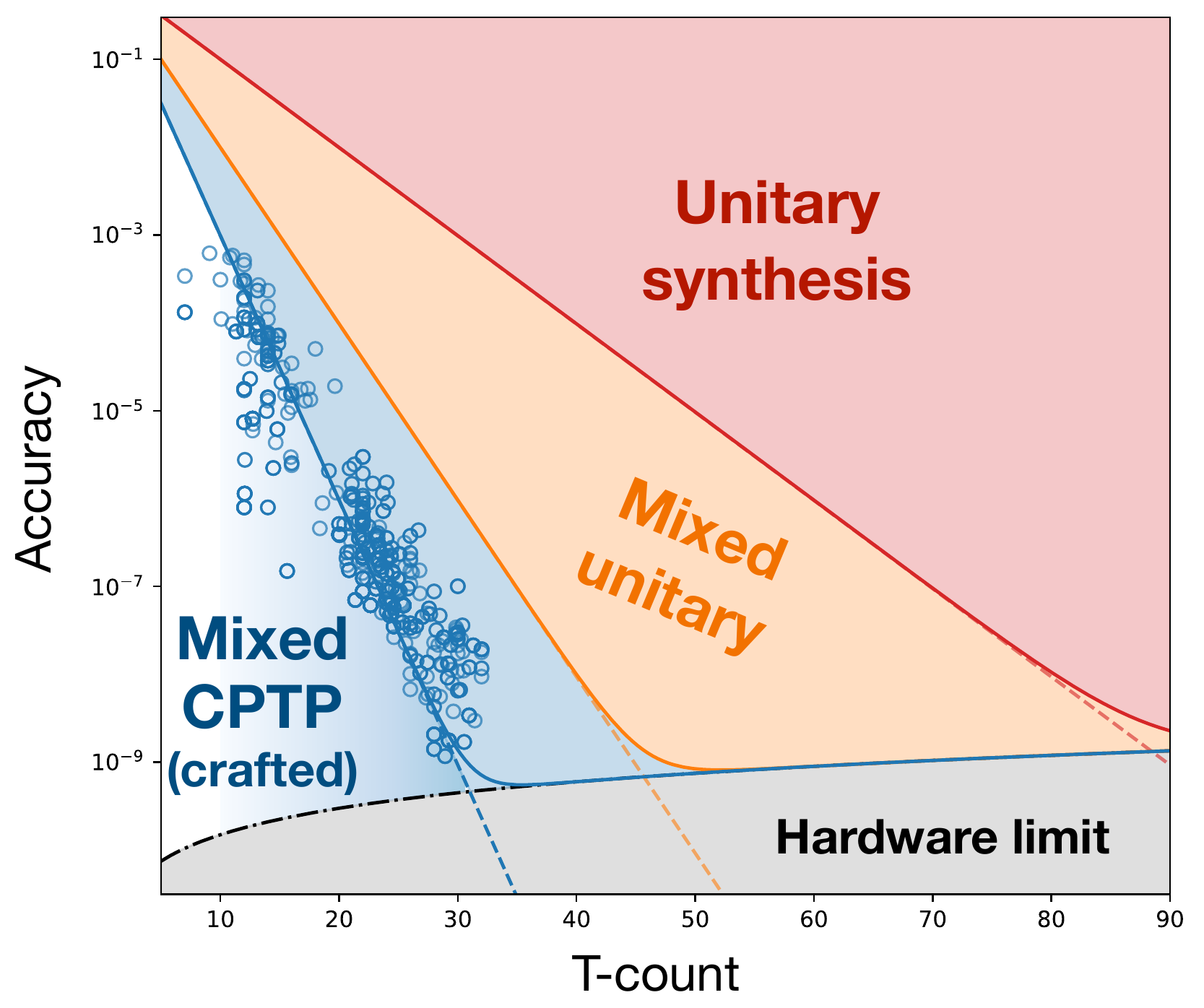}
    \caption{{\bf Trade-off between T-count and synthesis accuracy for Pauli rotation gates.} The red, yellow, and blue regions indicate the achievable accuracy without any sampling overhead when we employ unitary, mixed, and error-crafted synthesis.
    The blue \black{circles} indicate the error \black{solely from crafted synthesis for} $R_z(\theta)$ with $\theta\in \{\frac{m\pi}{32}\}_{m\neq 0~{\rm mod 4}}$ under various  synthesis accuracy. \black{Note that total error combined with hardware noise is bounded by the chained black line, which indicates the contribution from noise in T gates. Here we choose error rate of $10^{-11}$ per T gate.}
    \black{The dotted red, yellow, and green lines are given by the empirical T-count of $3 \log_2 (1/\varepsilon)$, $1.5 \log_2 (1/\varepsilon)$, and $\log_2 (1/\varepsilon)$, respectively,}
    \black{and the real lines indicate the total error.}
    %\black{The chained black line is given as a reference to indicate the contribution from noise in T gates of error rate $10^{-10}$. }
    }
    \label{fig:accuracy_tradeoff}
\end{figure}

\black{
Even if the Z error cannot be eliminated completely, 
such an error crafting drastically alters the sampling overhead. 
By performing the error detection operation as shown in Fig.~\ref{fig:detection_circuit}(a), we can mitigate Pauli X and Y errors with sampling overhead of $\gamma_{\rm XY}=1+p_x+p_y,$ which is even below the lower bound of QEM sampling overhead that does not rely on mid-circuit measurement~\cite{tsubouchiUniversalCostBound2023}. While the Z error is undetectable and hence we shall rely on other QEM methods such as PEC or rescaling with overhead of $\gamma_Z=1+\alpha p_z$, the total sampling overhead is $\gamma_{\rm tot}=\gamma_{\rm XY}\gamma_{Z}\sim1+p_x+p_y+\alpha p_z \sim 1+\epsilon^2$. When we perform QEM for $L$ gates, the total sampling overhead $\gamma_{\rm tot}^L$ is quartically small compared to that of PEC.
}

\black{An intuitive understanding of the reduction of the T-count can be obtained from a geometric interpretation. Let us first briefly describe the standard volume argument~\cite{harrow2002efficient} for SU(2) synthesis to yield scaling of $3 \log_2(1/\varepsilon)+O(1)$, and then proceed to explain the improved scaling of $\log_2(1/\varepsilon)+O(1)$. We can interpret the problem of designing a synthesis algorithm, i.e., ensuring any  $d$-dimensional unitary to be approximated by accuracy of $\varepsilon$, as packing $\epsilon$-ball into the manifold of SU($d$). The volume of $\varepsilon$-ball in $(d^2-1)$-dimensional space is $O(\varepsilon^{d^2-1}),$ in particular $V=O(\varepsilon^3)$ for SU(2). Hence, in order to cover the entire space, we need at least $O(1/V)=O(1/\varepsilon^3)$ unique unitaries. Note that any Clifford+T gate can be expressed uniquely by the canonical form proposed by  Matsumoto and Amano~\cite{matsumoto2008representation}. Specifically, the number of unique Clifford+T operator with T-count of $t$ is $192\cdot(3\cdot 2^t - 2)$. We therefore deduce that the lower bound of T-count scales as $3\log_2(1/\varepsilon) + O(1)$. It has been numerically verified that this bound is indeed nearly saturated by a direct search algorithm \cite{morisaki2024}, which reflects the fact that the grid points corresponding to Clifford+T unitary are distributed homogeneously within the manifold so that the gate set quickly mimics the Haar unitary ensemble~\cite{haferkamp2023efficient}.}

\black{
When we perform the mixed CPTP synthesis, since we can correct coherent error of X rotation via the feedback operation, we may allow it to be arbitrarily large. This effectively projects out one of the dimensions of the manifold to be covered. Correspondingly, the volume of the ball is regarded as $V=O(\varepsilon^2)$. From the volume argument, this leads to the T-count of $2 \log_2(1/\varepsilon) + O(1)$ for the synthesis of unitaries to be included in the mixture. By combining with the fact that appropriate mixed synthesis allows quadratic suppression of error and hence reduces the T-count by a factor of two, we arrive at the scaling of $\log_2(1/\varepsilon)+O(1)$.
}

\black{
Let us briefly comment on the classical precomputation time of crafted synthesis.
Our results in Fig.~\ref{fig:accuracy_tradeoff} exploits the direct search algorithm for SU(2) gate synthesis whose computational complexity scales as $O(1/\epsilon)$ on average~\cite{morisaki2024}. 
Considering that currently known unitary synthesis algorithms for $R_z$ gates based on number-theoretic approach~\cite{kliuchnikov2015practical, ross2016optimal} or repeat-until-success circuits~\cite{bocharov2015efficient, Kliuchnikov2023shorterquantum} yield nearly optimal gate sequence with \black{computational complexity} of $O({\rm polylog}(1/\epsilon))$, the complexity of direct search algorithms seems to be a significant overhead.
Meanwhile, we observe that the actual computation time per synthesis is observed to be milliseconds for the accuracy regime we have investigated.
Considering that such computations can be embarrassingly parallelized, and also that number of rotation gates for early FTQC circuits is of $10^4$ for quantum dynamics simulation~\cite{campbell2021early,beverland2022assessing} and $10^{6}$ for ground state energy simulation~\cite{kivlichan2020improved}, we envision that classical computation time is not bottleneck for such applications.
}

Last but not least, while we have exclusively discussed the error and sampling overhead only for the synthesis error, it must be kept in mind that, in reality, there is also incoherent noise that indirectly limits the synthesis accuracy.
As can be seen from Fig.~\ref{fig:accuracy_tradeoff},
one shall aim for the sweet spot where the total error is minimized.
While we have displayed the case when we have designed the FTQC to contribute the majority of logical error to magic state distillation, we may flexibly change the scheme to introduce errors in Clifford operations as well.

%\section{Summary and outlook}
\section*{Discussion}
In this work, we have proposed the scheme of error-crafted synthesis of single-qubit unitaries under the framework of Clifford+T formalism.
We have proven that we can guarantee the quadratic suppression of the unitary synthesis error even when we constrain the remnant error to be Pauli and depolarizing channel.
Furthermore, for the case of Pauli rotation gates, we have numerically shown that error crafting can yield biased Pauli noise such that sampling overhead to mitigate the errors can be suppressed quadratically compared to the existing techniques.
We have also found that error crafting for CPTP maps that incorporate measurement and feedback operations achieves cubic suppression of the error up to accuracy of $10^{-9}.$

Numerous future directions are envisioned.
First, it is essential to seek whether error crafting is possible for synthesis of multiqubit unitaries. 
Combined with the fact that the error can potentially be suppressed exponentially smaller with the qubit count for multiqubit case~\cite{akibue2024probabilistic}, it is crucial to extend the framework proposed in this work to perform practical quantum computation.
Second, it is intriguing in terms of resource theoretic argument to assess the probabilistic transformability between channels.
While this has been studied for state transformation in the context of magic state or entanglement distillation~\cite{regula2022probabilistic}, the argument on quantum channels is left as an interesting open question.
Third, it is an intriguing open question whether there is a lower bound on the computational complexity to perform error-crafted synthesis. 

\section*{Data availability}
Data and codes for the work are available via GitHub repository~\cite{crafting-github}.

\section*{Acknowledgement}
We thank fruitful discussions with Suguru Endo, Ali Javadi-Abhari, Yuuki Tokunaga, and Kunal Sharma.
This work was supported by 
JST CREST Grant Number JPMJCR23I4,  % Tokunaga-CREST
JST Moonshot R\&D Grant Number JPMJMS2061,
MEXT Q-LEAP Grant No. JPMXS0120319794, and No. JPMXS0118068682.
N.Y. wishes to thank JST PRESTO No. JPMJPR2119, % yoshioka,さきがけ
JST ERATO Grant Number  JPMJER2302,  % yoshioka, 沙川ERATO
JST Grant Number JPMJPF2221, % yoshioka, COI-NEXT
and the support from IBM Quantum. %yoshioka IBM
%S.E. is supported by JST PRESTO Grant Number JPMJPR2114.
S.A. is supported by JST, PRESTO Grant no.JPMJPR2111.
This work is supported by MEXT Quantum Leap Flagship Program (MEXT Q-LEAP) Grant No. JP- MXS0118067394 and JPMXS0120319794, JST COI- NEXT Grant No. JPMJPF2014.
\black{K.T. is supported by World-leading Innovative Graduate Study Program for Materials Research, Information, and Technology (MERIT-WINGS) of the University of Tokyo and JSPS KAKENHI Grant No. JP24KJ0857.}

\section*{Author contributions}
N.Y. conceived the idea. N.Y., S.A., and K.T. performed theoretical analysis. N.Y. and H.M. developed the numerical codes. All authors contributed to scientific discussion and writing of the paper.

\section*{Competing Interests}
Y.S. and S.A. own stock/options in Nippon Telegraph and Telephone. Y.S. owns stock/options in QunaSys Inc.

\bibliography{bib.bib}

\let\addcontentsline\oldaddcontentsline% Restore \addcontentsline

\clearpage
\onecolumngrid
\begin{center}
	\Large
	\textbf{Supplementary Notes for: Error Crafting in Mixed Quantum Gate Synthesis}
\end{center}

\setcounter{section}{0}
\setcounter{equation}{0}
\setcounter{figure}{0}
\setcounter{table}{0}
\setcounter{thm}{0}
\renewcommand{\thesection}{S\arabic{section}}
\renewcommand{\theequation}{S\arabic{equation}}
\renewcommand{\thefigure}{S\arabic{figure}}
\renewcommand{\thetable}{S\arabic{table}}
\renewcommand{\thethm}{S\arabic{thm}}
\setcounter{thm}{0}
\newtheorem{thmS}{Theorem S\ignorespaces}

\addtocontents{toc}{\protect\setcounter{tocdepth}{0}}

{
\hypersetup{linkcolor=blue}
\tableofcontents
}
\vspace{1cm}

\section{Limitation in applying twirling for non-Clifford operation}\label{sec:twirling-limitation}
\black{In the main text, we have explained using Table 1 that neither Pauli or Clifford twirling is not available for non-Clifford gates in general. Here we provide details on this argument.
}

\black{
Let $\mathcal{U}$ denote a non-Clifford single-qubit gate and $\mathcal{N}$ denote a noise channel associated to it.
When we wish to realize the implementation of the target unitary with twirled noise, its channel representation can be written as
\begin{eqnarray}\label{eq:twirl-ideal}
    \mathbb{E}_{\mathcal{D}} [\mathcal{D}^\dagger \circ \mathcal{N}\circ \mathcal{D}] \circ \mathcal{U},
\end{eqnarray}
where $\mathcal{D}$ is averaged over some single-qubit unitary subgroup such as the Pauli or Clifford group.
However, a gate operation is always accompanied by the noise as $\mathcal{N}\circ\mathcal{U}$, and we cannot insert $\mathcal{D}$ in between.
}

\black{
One of the options is to modify the expression as 
\begin{eqnarray}\label{eq:twirl-mod}
    \mathbb{E}_{D}  [\mathcal{D}^\dagger \circ \mathcal{N} \circ \mathcal{U} \circ (\mathcal{U}^\dagger \circ \mathcal{D} \circ \mathcal{U})],
\end{eqnarray}
%Observe that this indeed is equivalent to Eq.~\eqref{eq:twirl-ideal}.
which requires us to perform the following in sequence: 
\begin{enumerate}
    \item[(1)] Operate conjugated unitary $UDU^\dagger$ {\it without noise}.
    \item[(2)] Operate noisy target unitary.
    \item[(3)] Operate $D$ without noise. 
\end{enumerate}   
While the second and third steps are commonly considered to be achievable (coined as ``hard" and ``easy" operations in literature~\cite{wallman2016noise}), the problem is in the first step, since $U D U^\dagger $ is a non-Clifford operation in general. 
For instance, if we consider $U= e^{i \theta Z}$ and $D = X$, we must implement $UDU^\dagger = e^{i\theta Z} X e^{-i \theta Z} = e^{2 i \theta Z} X.$ Since we assume that rotation gates are noisy, this cannot be implemented without injecting additional untwirled noise. This is however completely contrary to what we desire to do.
}

\black{An exception is when the target gate belongs to the third level of the Clifford hierarchy. 
In this case, $U D U^\dagger$ is a Clifford operation if $D$ is a Pauli operator, and hence we can perform Pauli twirling by using Clifford gates. However, the Clifford twirling is not available;
if we take the Hadamard gate, we can show that
\begin{eqnarray}
    (T H T^\dagger) \cdot Z \cdot (T H T^\dagger)^\dagger = T H Z H T^\dagger = T X T^\dagger = SX.
\end{eqnarray}
This indicates that $THT^\dagger$ does not map a Pauli operator into a Pauli operator, and hence not a Clifford gate.
%does not cause any fundamental issue  as long as Clifford operations are significantly less noisy than the target gate. 
%This explains why the row of ``T gate" admits
}

\black{
One might consider another style of twirling by implementing $D^\dagger U D$ instead of $U D U^\dagger$. This direction is based on the fact that Eq.~\eqref{eq:twirl-ideal} can also be formally rewritten as 
\begin{eqnarray}
    \mathbb{E}_{\mathcal{D}} [\mathcal{D}^\dagger \circ \mathcal{N} \circ (\mathcal{D} \circ \mathcal{U} \circ \mathcal{D}^\dagger) \circ \mathcal{D}].
\end{eqnarray}
The requirement for this formulation is (1) the noise channel of $D^\dagger U D$ is all identically $\mathcal{N}$ irrespective of $\mathcal{D}$~\cite{wallman2016noise}, and (2) the noise in $\mathcal{D}$ is negligible.
When the target is $U=e^{i \theta P}$ with $P\in\{X, Y, Z\}$ and $\theta$ being the rotation angle, the requirement is that noise channel for Pauli rotation gates with angle $\theta$ are identical to each other.
However, this is unlikely in general, since the noise channel in unitary synthesis is unitary; to ensure the noise channel to be identical, we must accurately perform the unitary synthesis to control the synthesis error, which is beyond the actual task of synthesizing a target unitary up to accuracy of $\epsilon$.
}

%\black{
%We cannot perform the full Clifford twirling, we can instead consider the symmetric Clifford group whose elements all commutes with the target unitary [see Tsubouchi et al., arXiv:2405.07720].
%In this case the noise is only ``partially" twirled, in a sense that Z component of noise channel (overrotation or Pauli Z) is left as it is.
%}

\section{Formalism of mixed synthesis under diamond distance}\label{sec:formalism_prob_synth}

\black{In the main text, we have discussed the mixed synthesis that minimizes the diamond distance $\ddist(\Phi_1, \Phi_2)$ between two unitary channels $\Phi_1$ and $\Phi_2$.}
\black{
Here, we provide the formal definition of the diamond distance, which is related with the diamond norm as $\ddist(\Phi_1, \Phi_2):=\frac{1}{2}\|\Phi_1 - \Phi_2\|_{\diamond}$ where  $\Phi_1, \Phi_2$ are linear maps from $\mathcal{L}(\mathcal{H}_1)$ to $ \mathcal{L}(\mathcal{H}_2)$ ($(\mathcal{L}(\mathcal{H}_i))$ : the entire set of linear map on $\mathcal{H}_i$) for an $d$-dimensional Hilbert space $\mathcal{H}(=\mathcal{H}_1 = \mathcal{H}_2)$:}
\begin{eqnarray}
    \|\Phi_1 - \Phi_2 \|_{\diamond} := \max_{\rho \in \mathbb{D}(\mathcal{H}_1\otimes \mathcal{H}_2)} \left\| ((\Phi_1 - \Phi_2) \otimes {\rm Id}_d)(\rho) \right\|_1.
\end{eqnarray}
Here, $\|X \|_1 = {\rm Tr}[\sqrt{X^\dag X}]$ denotes the trace norm,  ${\rm Id}_d$ is an identity channel acting on $d$-dimensional Hilbert space, and $\mathbb{D}(\mathcal{H}_1 \otimes \mathcal{H}_2)$ is the set of physical density matrices on $\mathcal{H}_1 \otimes \mathcal{H}_2$. 
It is known that the diamond norm \black{$\|\Phi\|_{\diamond}$} can be efficiently calculated via semidefinite programming (SDP) for general linear maps as~\cite{watrous2009semidefinite, watrous2012simpler}
\begin{eqnarray}
\begin{tabular}{rlcrl}
\multicolumn{2}{c}{\underline{{\rm Primal problem}}} &\ \ \ \ \ \ \ \ \ \ \ \ \ \  
 &\multicolumn{2}{c}{\underline{{\rm Dual problem}}}\\
{\rm maximize:}&$\frac{1}{2}\langle \mathcal{J}(\Phi), X\rangle + \frac{1}{2} \langle\mathcal{J}(\Phi)^\dag, X^\dag \rangle$&&{\rm minimize:}& $\frac{1}{2} \|{\rm Tr}_{\mathcal{H}_2} Y_0\|_{\infty} + \frac{1}{2} \|{\rm Tr}_{\mathcal{H}_2} Y_1\|_\infty$ \\
{\rm subject to:}& $\begin{pmatrix}
        \rho_0 \otimes {\rm  Id}_d & X \\
        X^\dag & \rho_1 \otimes {\rm Id}_d
    \end{pmatrix} \geq 0$,&&
    {\rm subject~to~}: &$\begin{pmatrix}
        Y_0 & - \mathcal{J}(\Phi) \\
        -\mathcal{J}(\Phi)^\dagger & Y_1
    \end{pmatrix} \geq 0$.\\
&$\rho_0, \rho_1\in\mathbb{D}(\mathcal{H}_1)$, $X \in \mathcal{L}(\mathcal{H}_1 \otimes \mathcal{H}_2)$. &&& $Y_0, Y_1 \in \mathcal{L}(\mathcal{H}_1) \otimes \mathcal{L}(\mathcal{H}_2)$, $Y_0=Y_0^\dag, Y_1=Y_1^\dag$.
%\\
%&  &&&.
\end{tabular}     
\end{eqnarray}
where $\langle P, Q\rangle = {\rm Tr}[P^\dag Q]$ is the inner product of operators $P$ and $Q$,  $Y_0, Y_1 \in \mathcal{L}(\mathcal{H}_1) \otimes \mathcal{L}(\mathcal{H}_2)$ are Hermitian matrices, $\|\cdot\|_{\infty}$ is the spectral norm, and $\mathcal{J}(\Phi)$ is the Choi operator of a map $\Phi$ defined as $\mathcal{J}(\Phi) := ({\rm Id}_d \otimes \Phi)[|\Omega\rangle \langle \Omega|]$
for the maximally entangled state $|\Omega \rangle = \sum_ i |i\rangle_{\mathcal{H}_1}|i\rangle_{\mathcal{H}_2}$.
\black{By further assuming that $\Phi_1$ and $\Phi_2$ are completely positive and trace preserving (CPTP) maps, the calculation can be shown to be simplified as follows.}
\begin{eqnarray}
\begin{tabular}{rlcrl}
\multicolumn{2}{c}{\underline{{\rm Primal problem}}} &\ \ \ \ \ \ \ \ \ \ \ \ \ \  
 &\multicolumn{2}{c}{\underline{{\rm Dual problem}}}\\
{\rm maximize:}&${\rm Tr}[\mathcal{J}(\Phi_1 - \Phi_2) T]$&&{\rm minimize:}& $r \in \mathbb{R}$ \\
{\rm subject to:}& $0 \leq T \leq \rho \otimes {\rm Id}_d$,&&
    {\rm subject~to~}:& $S\geq0, S\geq \mathcal{J}(\Phi_1 - \Phi_2),$\\
& $\rho\in\mathbb{D}(\mathcal{H}_1).$ &&& $r {\rm Id}_{d} \geq {\rm Tr}_{{\mathcal{H}_2}}[S]$.
\end{tabular}     
\end{eqnarray}

Now let us consider  mixed synthesis of a target CPTP map $\hat{\Upsilon}$ using a  set of CPTP implementable maps $\{\Upsilon_j\}$ under the following minimization problem:
\begin{eqnarray}\label{eq:prob_synthesis}
\begin{tabular}{c}
    {\rm minimize:}~$\frac{1}{2} \left \| \hat{\Upsilon} - \sum_j  p_j \Upsilon_j \right\|_{\diamond}$\\
    {\rm subject~to:}~$p_j \geq 0, \sum_j p_j = 1$.
\end{tabular}    
\end{eqnarray}
While it seems complicated since the diamond norm is already defined by minimization, it can be shown that we can solve Eq.~\eqref{eq:prob_synthesis} via SDP:
\begin{proposition}
\label{prop:SDP}
Let $\hat{\Upsilon}$ and $\{\Upsilon_j\}$ be a target CPTP map and a finite set of implementable CPTP maps from $\mathcal{L}(\mathcal{H}_1)$ to $\mathcal{L}(\mathcal{H}_2)$, respectively. Then, the distance $\min_{p}\frac{1}{2}\diamondnorm{\hat{\Upsilon}-\sum_{j} p_j \Upsilon_j}$ and the optimal probability distribution $\{p_j\}$ can be computed with the following SDP:
\begin{equation}
\label{eq:SDP}
\begin{tabular}{rlcrl}
\multicolumn{2}{c}{\underline{{\rm Primal problem}}} &\ \ \ \ \ \ \ \ \ \ \ \ \ \  
 &\multicolumn{2}{c}{\underline{{\rm Dual problem}}}\\
{\rm maximize:}&${\rm Tr}[\mathcal{J}(\hat{\Upsilon})T]-t$&&{\rm minimize:}&$w\in\mathbb{R}$ \\
{\rm subject to:}& $0\leq T\leq\rho\otimes {\rm Id}_{\mathcal{H}_2}$,&&
{\rm subject to:}& $S\geq0,  S\geq \mathcal{J}\left(\hat{\Upsilon} - \sum_j p_j \Upsilon_j\right)$,\\
&$\rho\in\mathcal{D}(\mathcal{H}_1)$, &&&$w{\rm Id}_{\mathcal{H}_1}\geq{\rm Tr}_{\mathcal{H}_2}[S]$,\\
&$\forall j,{\rm Tr}[\mathcal{J}(\Upsilon_j)T]\leq t$.&&&$\forall p_j\geq0$, $\sum_j p_j\leq 1$.
\end{tabular} 
\end{equation}
Note that the strong duality holds in this SDP, i.e., the optimum primal and dual values are equal.
\end{proposition}
%When we consider CPTP maps, the diamond distance can be shown to be equivalent to 

\subsection{Simplification of mixed synthesis in single-qubit case}\label{app:sdp_simple}
Although SDP is efficient in a sense that there exists a polynomial-time algorithm with respect to the problem matrix size, in practice the runtime grows quite rapidly, and hence it is desirable to simplify the minimization problem.
For the single-qubit case, the calculation of diamond distance boils down to that of trace distance between Choi matrices.
While one can show \black{$\frac{1}{d}\|\Phi_1 - \Phi_2\|_{\diamond} \leq \frac{1}{d}\|\mathcal{J}(\Phi_1) - \mathcal{J}(\Phi_2)\|_{1} \leq \|\Phi_1 - \Phi_2\|_{\diamond}$} for general $d$-dimensional completely positive maps $\Phi_1$ and $\Phi_2$, for the single-qubit \black{mixed unitary channel} case it is known to saturate the lower bound as $\|\Phi_1 - \Phi_2\|_{\diamond} = \|\mathcal{J}(\Phi_1) - \mathcal{J}(\Phi_2)\|_1/2.$
Therefore, the remnant error of mixed synthesis, $\min_{p}\frac{1}{2}\diamondnorm{\hat{\Upsilon}-\sum_{j}p_j\Upsilon_j}$, can now be computed by the following simpler semidefinite programming problem:
\begin{equation}
\label{eq:1QSDP}
\begin{tabular}{rlcrl}
\multicolumn{2}{c}{\underline{{\rm Primal problem}}} &\ \ \ \ \ \ \ \ \ \ \ \ \ \  
 &\multicolumn{2}{c}{\underline{{\rm Dual problem}}}\\
{\rm maximize:}&$\frac{1}{2}{\rm Tr}[M \mathJ(\hat{\Upsilon})]-t$&&{\rm minimize:}&${\rm Tr}[Y]$ \\
{\rm subject to:}& $0\leq M\leq\idop$,&&
{\rm subject to:}& $Y\geq0, Y\geq \frac{1}{2}\mathJ(\hat{\Upsilon})-\frac{1}{2}\sum_{j}p_j \mathJ(\Upsilon_j)$,\\
&$\forall j,\frac{1}{2}{\rm Tr}[M\mathJ(\Upsilon_j)]\leq t$.&&&$\forall j,p_j \geq0$,$\sum_{j} p_j\leq 1$.
\end{tabular} 
\end{equation}

We further find that it is convenient to introduce the notion of magic basis, an orthonormal basis for the Choi representation of single-qubit channels: 
\begin{eqnarray}
\label{eq:MESbasis}
\ket{\Psi_1}=\frac{1}{\sqrt{2}}(\ket{00}+\ket{11}),\ \ \ &&\ket{\Psi_2}=\frac{i}{\sqrt{2}}(\ket{00}-\ket{11}),\nonumber\\
\ket{\Psi_3}=\frac{i}{\sqrt{2}}(\ket{01}+\ket{10}),\ \ \ &&\ket{\Psi_4}=\frac{1}{\sqrt{2}}(\ket{01}-\ket{10}).
\end{eqnarray}
It can be shown that this relates a single-qubit unitary channel $\mathcal{U}$ to a unit vector $r=(r_1, r_2, r_3, r_4)^T\in\rr^4$ as follows,
\begin{equation}
\label{eq:MBrep}
 \frac{1}{2}\left[ \mathcal{J}(\mathcal{U})\right]_{MB}=rr^T,
\end{equation}
where $\left([\cdot]_{MB}\right)_{ij} = \langle \Psi_i | \cdot | \Psi_j \rangle$.

\black{When we explicitly consider the distance between a target unitary $\hat{\mathcal{U}}$ and a mixture of unitaries $\sum_j p_j \mathcal{U}_j$, we can rewrite the expression of the diamond distance as 
\begin{eqnarray}
    d_{\diamond}\left(\hat{\mathcal{U}}, \sum_j p_j \mathcal{U}_j\right) = \frac{1}{2}\left\| \hat{\mathcal{U}} - \sum_j p_j \mathcal{U}_j \right\|_{\diamond} = 
    \frac{1}{2}\cdot \frac{1}{2} \left\|\mathcal{J}(\mathcal{U}) - \sum_j p_j \mathcal{J}(\mathcal{U}_j)\right\|_1  = 
    \frac{1}{2} \left\|rr^T - \sum_j p_j r^{(j)} r^{(j)T}\right\|_1,
\end{eqnarray}
where $r$ and $r^{(j)}$ are the unit vectors obtained from the magic basis representation of $\hat{\mathcal{U}}$ and $\mathcal{U}_j$, respectively.}
%If $\hat{U}$ is an identity, we have $r=(1,0,0,0)$. Therefore, a
\black{
%We remark that, the diamond distance between two Pauli channels  $\Phi_1 = \sum_i p_i \mathcal{P}_i$ and $\Phi_2 = \sum_i q_i \mathcal{P}_i$  can be computed simply as $\| \Phi_1 - \Phi_2\|_{\diamond} = \sum_i |p_i - q_i|$ for arbitrary qubit count. 
Furthermore, since the trace norm $\|A\|_1$ is simply a sum of diagonal elements when $A$ is diagonal, the diamond distance can be simplified 
when $\hat{\mathcal{U}}$ is identity and $\sum_j p_j \mathcal{U}_j$ is a single-qubit Pauli channel as
\begin{eqnarray}
    d_{\diamond}\left(\hat{\mathcal{U}}, \sum_j p_j \mathcal{U}_j \right) = \frac{1}{2} \left((1 - \sum_{j} p_j (r_1^{(j)})^2) + \sum_{k=2}^4 \sum_j p_j (r_{k}^{(j)})^2 \right) = 1 - \sum_j p_j (r_1^{(j)})^2,
\end{eqnarray}
where we have used $\sum_{k=2}^4 (r_k^{(j)})^2 = 1- (r_1^{(j)})^2$ for the second equation.
As we discuss in Sec.~\ref{sec:numerics_detail}, this simplification is beneficial since the minimization problem of the diamond distance can be formulated by linear programming problem instead of the cumbersome SDP.
}

\section{Numerical details on error crafting using unitaries}\label{sec:numerics_detail}

\subsection{Crafting Pauli channel}\label{subsec:numerics_pauli_constraint}
As discussed in the main text, the minimization problem regarding the error-crafted synthesis of a Pauli channel is formulated via the following constrained minimization problem of a probability distribution $\{p_j\}$ given a target unitary channel $\hat{\mathU}$ and a set of synthesized unitaries $\{\mathU_j\}$ as follows:
\begin{equation}
\label{eq:Paulierror}
{\rm minimize}\ 
\frac{1}{2}\diamondnorm{\hat{\mathU}-\sum_j p_j\mathU_j}\ \ \ {\rm subject\ to}\ \ \sum_j p_j \mathU_j\circ\hat{\mathU}^{-1}(\cdot)=\sum_{a}\chi_{aa} P_a\cdot P_a, \chi_{aa}\geq0,
\end{equation}
where $P_a \in \{{\rm I},{\rm X},{\rm Y}, {\rm Z}\}$ is a single-qubit Pauli operator. 
In a practical sense, an equality constraint tends to yield numerical instability, and therefore we relax the condition for the remnant error $\mathcal{E}_{\rm rem} = (\sum_j p_j \mathU_j) \circ \hat{\mathU}^{-1}$ to allow a small violation as 
\begin{equation}
\label{eq:Paulierror_relax}
{\rm minimize}\ 
\frac{1}{2}\diamondnorm{\hat{\mathU}-\sum_j p_j\mathU_j}\ \ \ {\rm subject\ to}\ \ \sum_j p_j \mathU_j\circ\hat{\mathU}^{-1}(\cdot)=\sum_{ab}\chi_{ab} P_a\cdot P_b, \sum_{a\neq b}|\chi_{ab}| \leq g_{1},
\end{equation}
where $g_1$ is a tolerance factor regarding the Pauli error constraint.

When the set $\{\mathU_j\}$ is not appropriately generated, there does not exist any feasible solution for Eq.~\eqref{eq:Paulierror}.
While numerical solvers such as Gurobi offer an automated way to relax the equation constraint, via the variable {\tt presolve}, we find that this precomputation is done too aggressively so that the accuracy of the minimization is quite poor.
Instead, in our work, we have done a logarithmic search of $g_1$, starting from machine precision of $g_1 = 10^{-16}$.
We have defined the error crafting to be successful if we find feasible solution such that (1)  $\min \frac{1}{2} \|\hat{\mathU} - \sum_j p_j \mathU_j\|_{\diamond} \leq (c+1)^2\epsilon^2,$ and (2) $g_{1}\leq 10^{-12}$ are both satisfied.

As we discussed in Sec.~\ref{sec:formalism_prob_synth}, the minimization problem boils down to a linear programming problem under the magic basis. 
\black{The minimization problem equivalent to Eq.~\eqref{eq:Paulierror} is 
\begin{equation}
\label{eq:magic_basis_LP_pauli}
 {\rm minimize}\ 1-\sum_j p_j\left(r^{(j)}_1\right)^2\ \ \ {\rm subject\ to}\ \ \sum_{j} p_j = 1, p_j \geq 0, \sum_j p_j r^{(j)} r^{(j)T} {\rm \ is\ diagonal,}
\end{equation}
where $r^{(j)}=(r_1^{(j)},r_2^{(j)},r_3^{(j)},r_4^{(j)})^T\in\rr^4$ corresponds to the magic basis representation of $\mathU_j\circ\hat{\mathU}^{-1}.$}
\black{After relaxing the equality constraint, the minimization problem is given as
\begin{equation}
\label{eq:magic_basis_LP_pauli_relaxed}
 {\rm minimize}\ 1-\sum_j p_j\left(r^{(j)}_1\right)^2\ \ \ {\rm subject\ to}\ \ \sum_{j} p_j = 1, p_j \geq 0, \sum_{k\neq k'}\left|\sum_j p_j r^{(j)}_k r^{(j)}_{k'}\right| \leq g'_1,
\end{equation}
 \black{where $g'_1$ is the tolerance factor in this representation}. 
 }
We find that accuracy is much enhanced when the minimization (and relaxation) is done under this representation.
In practice, we solve Eq.~\eqref{eq:magic_basis_LP_pauli_relaxed} using LP implemented in Gurobi or SciPy, and check if the conditions are satisfied.

\subsection{Crafting depolarizing channel}\label{subsec:numerics_depol_constraint}
Here, we provide numerical details regarding the error crafting of depolarizing channels.
In similar to Sec.~\ref{subsec:numerics_pauli_constraint}, the minimization problem is defined as
\begin{eqnarray}
\label{eq:Depolerror}
{\rm minimize}\ 
\frac{1}{2}\diamondnorm{\hat{\mathU}-\sum_j p_j\mathU_j}\ \ \ {\rm subject\ to}\ \ 
\sum_{a} \chi_{aa}=1, \chi_{\rm XX} = \chi_{\rm YY} = \chi_{\rm ZZ}\geq 0, \chi_{a \neq b} = 0,
\end{eqnarray}
where the set of shift unitary $\{\mathcal{V}_{j}^{(c \epsilon)}\}_{j=1}^9$ consists of nine elements, in contrast to seven for the Pauli constraint.
Again we relax the equality constraints into inequality constraints for the sake of numerical stability.
Concretely, we solve the following problem:
\begin{equation}
\label{eq:Depolerror_relaxed}
{\rm minimize}\ 
\frac{1}{2}\diamondnorm{\hat{\mathU}-\sum_j p_j\mathU_j}\ \ \ {\rm subject\ to}\ \ 
\sum_a \chi_{aa}=1, 
\max_{a'\neq b'} \{|\chi_{a'a'} - \chi_{b'b'}|\} \leq g_2,
 \sum_{a \neq b}|\chi_{ab}| \leq g_1,
\end{equation}
where $a \in \{{\rm I}, {\rm X}, {\rm Y}, {\rm Z}\}$ and $a', b' \in \{{\rm X}, {\rm Y}, {\rm Z}\}.$
We regard the error crafting to be successful if  (1)  $\min \frac{1}{2} \|\hat{\mathU} - \sum_j p_j \mathU_j\|_{\diamond} \leq (c+1)^2\epsilon^2$, (2) $g_1 \leq 10^{-12}$, and (3) $g_2 \leq 0.01 (\sum_{a'\in\{X,Y,Z\}} \chi_{a'a'}) \sim 0.01 \epsilon^2$ are all satisfied.
\black{For completeness, we mention that the minimization problem can also be expressed in the magic basis representation as 
\begin{eqnarray}
     && {\rm minimize}\ 1-\sum_j p_j\left(r^{(j)}_1\right)^2\ {\rm subject\ to}\ \ \sum_{j} p_j = 1, p_j \geq 0, \ \ \sum_j p_j r^{(j)}r^{(j)T } = (1-q){\mathbf e}^{(1)}{\mathbf e}^{(1)T} + \frac{q}{4}{\rm Id},
\end{eqnarray}
where ${\mathbf e}^{(i)}$ is a vector whose element is 0 except the $i$th element.
After relaxing the equality into inequality, we now have
\begin{eqnarray}
     && {\rm minimize}\ 1-\sum_j p_j\left(r^{(j)}_1\right)^2\ \ \ \nonumber\\
     && {\rm subject\ to}\ \ \sum_{j} p_j = 1, p_j \geq 0, \ \ 
     \max_{k \neq k'}\left\{\left|\sum_j p_j (r_k^{(j)} r_{k}^{(j)} - r_{k'}^{(j)}r_{k'}^{(j)})\right|\right\}\leq g'_2,\ \ 
     \sum_{k\neq k'}\left|\sum_j p_j r^{(j)}_k r^{(j)}_{k'}\right| \leq g'_1,
\end{eqnarray}
where $g'_2$ is the tolerance factor for the homogeneity of the Pauli error rates.
}

The results of error-crafted synthesis are shown in Fig.~\ref{fig:dnorm_quadratic_depol}.
We observe the validity of the nine-shift unitary from the bounded behavior in Fig.~\ref{fig:dnorm_quadratic_depol}(a). As in the case of the Pauli constraint, we also observe a boost in the accuracy by expanding the size of the basis set by a factor of $R$.
From Figs.~\ref{fig:dnorm_quadratic_depol}(b), (c), and (d), we further find that both increasing $c$ and $R$ contribute to the suppression of the failure probability of the error crafting.

\begin{figure}[h]
    \centering
    \includegraphics[width=0.98\linewidth]{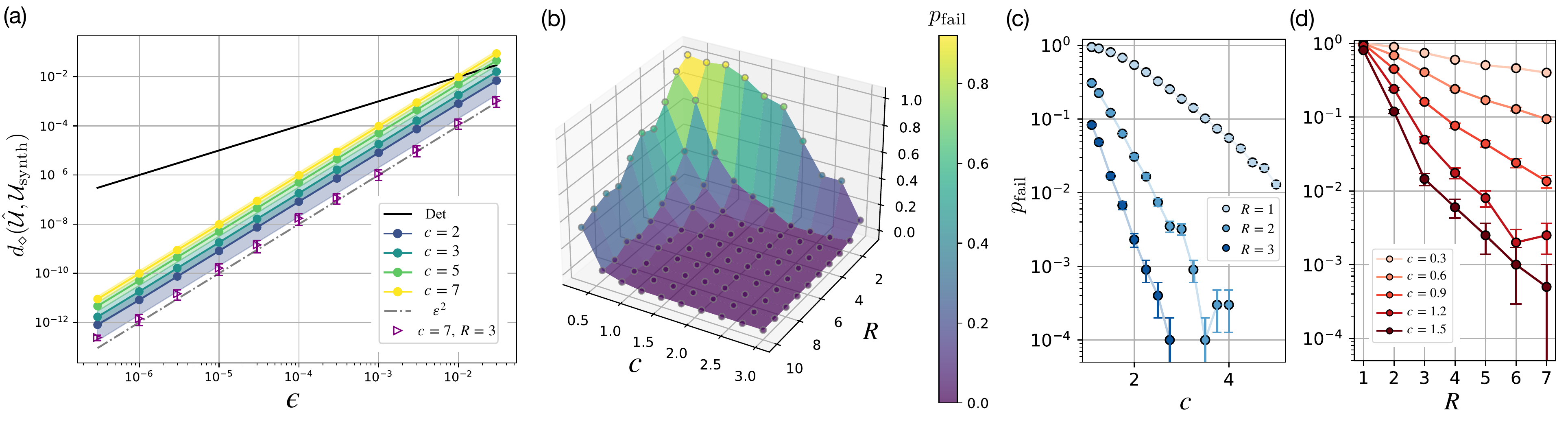}
    \caption{
    (a) Diamond distance $\ddist(\hat{\mathU}, \mathU_{\rm synth})$ between the target single-qubit Haar random unitary $\hat{\mathU}$ and synthesized channel $\mathU_{\rm synth} \coloneqq\sum_j p_j \mathU_j$, with the remnant error constrained to the single-qubit depolarizing channel. The black real line indicates the results from unitary synthesis by the Ross-Selinger algorithm~\cite{ross2016optimal}, while the filled dots are results from error-crafted synthesis with shift factors of $c=2, 3, 5, 7.$ When we expand the basis set size by a factor of $R=3$, shown by unfilled triangles, the effective synthesis error is suppressed below the bound in Eq.~\eqref{eq:proberrorDepo} and approaches the value of $\epsilon^2$.
    The shown data are mean over 200 \black{random instances of target unitaries}.
    (b) Failure rate of error-crafted synthesis, averaged over 50 \black{instances}.
    %The inset shows similar plot for enlarged region.
    (c)(d) Suppression of $p_{\rm fail}$ with shift factor $c$ and $R$. The data is averaged over 10000 and 2000 \black{instances}, respectively.
    }
    \label{fig:dnorm_quadratic_depol}
\end{figure}

\subsection{Crafting XY channel} \label{subsec:XY_constraint}
\black{We finally describe the minimization problem to craft XY channels.
Namely, we wish to solve the following:
\begin{equation}
\label{eq:Depolerror}
{\rm minimize}\ 
\frac{1}{2}\diamondnorm{\hat{\mathU}-\sum_j p_j\mathU_j}\ \ \ {\rm subject\ to}\ \ 
\chi_{XX}+\chi_{YY}=1, \chi_{\rm ZZ} = 0, \chi_{a \neq b} = 0,
\end{equation}
where the set of shift unitary $\{\mathcal{V}_j^{(c \epsilon)}\}$ is generated according to the procedure shown in Sec.~\ref{sec:shift-unitary-gen}.
To enhance the numerical stability, we have relaxed the equality constraint for $\chi_{ZZ}$ into inequality as
\begin{equation}
\label{eq:Depolerror_relaxed}
{\rm minimize}\ 
\frac{1}{2}\diamondnorm{\hat{\mathU}-\sum_j p_j\mathU_j}\ \ \ {\rm subject\ to}\ \ 
\sum_a \chi_{a}=1, 
\max_{a'\neq b'} \{|\chi_{ZZ}|\} \leq g_2,
 \sum_{a \neq b}|\chi_{a \neq b}| \leq g_1,
\end{equation}
where tolerance factors $g_1$ and $g_2$ are chosen in similar to those for depolarizing channels.
Using the magic basis, the minimization problem is expressed as 
\begin{eqnarray}
     && {\rm minimize}\ 1-\sum_j p_j\left(r^{(j)}_1\right)^2\ \ \ \nonumber\\
     && {\rm subject\ to}\ \ \sum_{j} p_j = 1, p_j \geq 0, \ \ 
     \sum_j p_j r_2^{(j)} r_{2}^{(j)}=0,\ \ 
     \sum_j p_j r^{(j)}r^{(j)T} {\rm \ is \ diagonal},
\end{eqnarray}
and after the relaxation of the equality constraint into inequality we have
\begin{eqnarray}
     && {\rm minimize}\ 1-\sum_j p_j\left(r^{(j)}_1\right)^2\ \ \ \nonumber\\
     && {\rm subject\ to}\ \ \sum_{j} p_j = 1, p_j \geq 0, \ \ 
     \left|\sum_j p_j r_2^{(j)} r_{2}^{(j)}\right|\leq g'_2,\ \ 
     \sum_{k\neq k'}\left|\sum_j p_j r^{(j)}_k r^{(j)}_{k'}\right| \leq g'_1,
\end{eqnarray}
where $g'_2$ is the tolerance factor for suppressing the Z component of the Pauli channel and $g'_1$ is for non-unital components.
}

%\subsection{Comparison between crafted and uncrafted  mixed synthesis} \label{subsec:crafted-vs-noncrafted}

%\begin{figure}[b]
    %\centering
    %\includegraphics[width=0.95\linewidth]{figures/crafted_vs_non-crafted.pdf}
%    \caption{
%    Crafted versus non crafted
%    }
%    \label{fig:crafted-vs-noncrafted}
%\end{figure}

%\black{Here we discuss...}

\begin{figure}[h]
    \centering
    \includegraphics[width=0.38\linewidth]{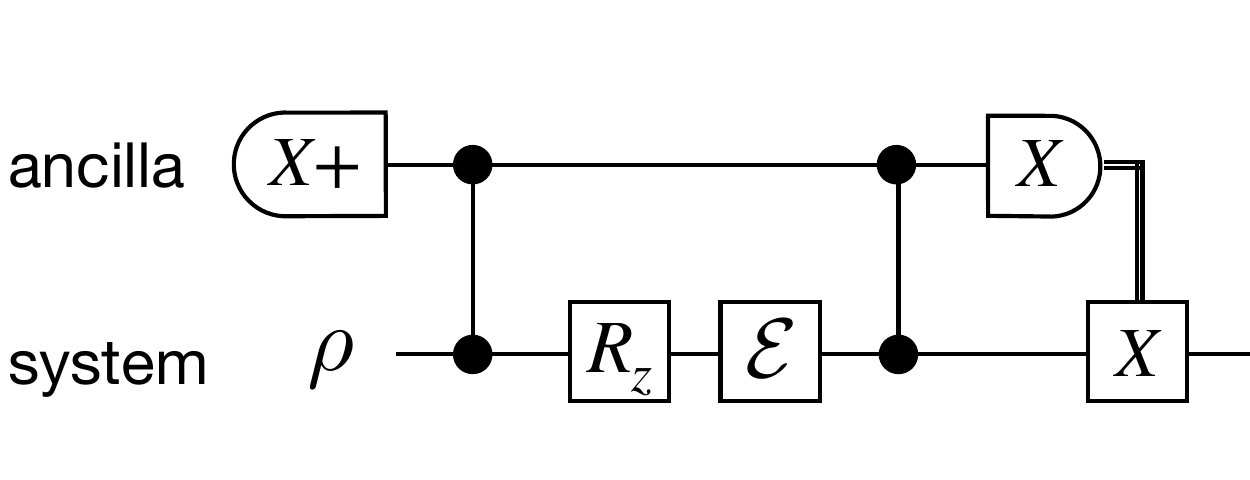}
    \caption{
    Measurement and feedback protocol that corrects X error in implementing a Pauli Z rotation. 
    %The figure is identical to Fig.~\ref{fig:detection_circuit}(b) in the main text.
    }
    \label{fig:detection-circuit-SM}
\end{figure}

\section{Numerical details on error crafting using CPTP maps} \label{sec:craft-CPTP}
As discussed in the main text, the error of Pauli rotation can be suppressed significantly by applying the error crafting for CPTP maps that represents the channel of synthesized gates accompanied by error correction circuit as shown in Fig.~\ref{fig:detection-circuit-SM}.
\black{In similar to the cases in Sec.~\ref{sec:numerics_detail}, the minimization problem is defined as
\begin{eqnarray}
\frac{1}{2}\diamondnorm{\hat{\mathU}-\sum_j p_j\Lambda_j}\ \ \ {\rm subject\ to}\ \ \sum_j p_j \Lambda_j\circ\hat{\mathU}^{-1}(\cdot)=\sum_{a}\chi_{aa} P_a\cdot P_a, \chi_{aa}\geq0,    
\end{eqnarray}
where $\{\Lambda_j\}$ is the set of implementable CPTP maps.}
\black{Concretely for the case of Pauli rotation, with} the target unitary and synthesized unitary channels expressed as $\hat{\mathU}$ and $\mathU_j$, respectively, we can express the effective CPTP map $\Lambda_j$ and remnant error  
$\mathcal{E}_{{\rm rem}, j}$ obtained from logical error correction as
\begin{eqnarray}
\Lambda_j &=& \mathscr{T}\left(\mathcal{C} \circ \mathcal{M}_{X}^{(A)} \circ \mathcal{W}_{E}' \circ \mathcal{U}_j  \circ \mathcal{W}_E \circ (\Pi_{X+}^{(A)})\right), \\    
    \mathcal{E}_{{\rm rem},j} &=& \mathscr{T}\left(\mathcal{C} \circ \mathcal{M}_{X}^{(A)} \circ \mathcal{W}_{E}' \circ (\mathcal{U}_j \circ \hat{\mathcal{U}}^{-1}) \circ \mathcal{W}_E \circ (\Pi_{X+}^{(A)})\right),\label{eq:Erem_j}
\end{eqnarray}
where $\mathscr{T}$ trace outs the ancillary qubit, $\mathcal{C}$ is the correction operation conditioned by the result of the X-measurement on ancilla  $\mathcal{M}_{X}^{(A)}$, $\mathcal{W}_E,\mathcal{W}'_E$ is the entangling operation between the target and ancilla qubit, and $\Pi_{X+}^{(A)}$ is the state preparation (or force initialization) of $|+\rangle$ state on the ancilla.
Note that the expression of Eq.~\eqref{eq:Erem_j} follows from the fact that $\mathcal{W}_E$ and $\mathcal{W}'_E$ commutes with the target unitary $\hat{\mathcal{U}}.$

Without loss of generality, we consider the single-qubit Pauli Z rotation in the following.
First, note that, for the case of Pauli rotations, the goal of finding an appropriate set of CPTP maps and their weights $\{(p_j, \Lambda_j)\}_j$ is equivalent to finding an appropriate set of $\{p_j, \mathcal{E}_{{\rm rem}, j}\}_j$.
It is straightforward to see that the Pauli transfer matrix of a remnant error channel $\mathcal{E}_{{\rm rem}, j}$ for a Pauli Z rotation, denoted as $ \Gamma_{\mathcal{E}_{{\rm rem}, j}}$, can be expressed as 
\begin{eqnarray}
\Gamma_{\mathcal{E}_{{\rm rem}, j}} = 
    \begin{pmatrix}
        1 & 0 & 0 & 0 \\
         0& 1-\mu_j & \nu_j & 0 \\
         0& -\nu_j & 1-\mu_j & 0 \\
         0& 0 & 0 & 1 
    \end{pmatrix},
\end{eqnarray}
where $0 \leq \mu_j\leq 1$ and $\Gamma_{\mathcal{E}_{{\rm rem}, j}} \geq 0.$
This motivates us to find a pair of remnant error channels $\mathcal{E}_{{\rm rem}, +}$ and $\mathcal{E}_{{\rm rem},-}$ which satisfies (i) $\nu_+ > 0, \nu_- < 0$, (ii) $\mu_{\pm} = O(\epsilon^3)$\black{, which follows from the geometric interpretation described in the main text.}
Once such a pair is found, we take the mixture of them as
\begin{eqnarray}
    \tilde{\mathcal{E}}_{\rm rem} &=& \frac{\nu_-}{\nu_+ + \nu_-}\mathcal{E}_{{\rm rem}, +} + \frac{\nu_+}{\nu_+ + \nu_-}\mathcal{E}_{{\rm rem}, -}, \\
    \Gamma_{\tilde{\mathcal{E}}_{\rm rem}} &=& \begin{pmatrix}
        1 & 0 & 0 & 0 \\
         0& 1-\tilde{\mu} & 0 & 0 \\
         0& 0 & 1-\tilde{\mu} & 0 \\
         0& 0 & 0 & 1         
    \end{pmatrix},
\end{eqnarray}
which is a purely Z error channel with an error rate of $\tilde{\mu} = \frac{\nu_-}{\nu_+ + \nu_-} \mu_+ + \frac{\nu_+}{\nu_+ + \nu_-}\mu_-=O(\epsilon^3).$ 

We note that the current search algorithm is based on a brute-force search, and hence the computational complexity is expected to scale as  $O(K/\epsilon)$ where $K$ is the number of shift unitaries. We leave an open question of whether there is any bound on complexity to search the pair as in the above discussion. In the current implementation, we generate around 3000 unitaries to reach a remnant error of $10^{-6}$ and around 20000 unitaries to reach $10^{-9},$ with each unitary synthesis done in milliseconds using a direct search algorithm proposed in Ref.~\cite{morisaki2024}.
Note that each unitary synthesis is independent, and hence is embarrassingly parallelizable.

\section{Generation of shift unitaries with fixed diamond distance} \label{sec:shift-unitary-gen}
In this section, we discuss how to generate a set of 
 $S$ shift unitaries $\{\mathcal{V}^{\epsilon}_s\}_{s=1}^S$ which all satisfies $\ddist(\mathcal{V}_s^{\epsilon}, \mathcal{I})= \epsilon.$
 One of such a method is to generate Haar random unitaries and constrain the ``rotation angle" by noting that any single-qubit unitary can be written as $e^{i \theta \sum_{a=X,Y,Z} n_a P_a}$ where $(n_X, n_Y, n_Z)^T\in\mathbb{R}^3$ is a unit vector. 
 However, we have found that such a protocol is unsatisfactory in terms of efficiency in generating independent unitaries after unitary synthesis.
 Therefore, we instead utilize the magic basis representation of a single-qubit unitary. To be precise, given a unit vector $\tilde{v}\in \mathbb{R}^3,$ we can generate a single-qubit unitary that is expressed in the magic basis representation as
\begin{eqnarray}
    e e^T = \frac{1}{2}\left[ \mathcal{J}(\mathcal{U}) \right]_{MB}{\rm~with~}e=\begin{pmatrix}
        \sqrt{1-\epsilon^2} \\
        \epsilon \tilde{v}
    \end{pmatrix}.\label{eq:mb_choivec}
\end{eqnarray}
Therefore, we may generate unit vectors in order to generate a set of shift unitaries with a fixed value of the diamond distance as follows:
\begin{enumerate}
    \item[] Step 1. Generate a set of unit vectors $\{\tilde{v}_s | \tilde{v}_s\in \mathbb{R}^3, \|\tilde{v}_s\|_2=1\}_{s=1}^S$ so that the vectors are distributed on the surface of the sphere as homogeneously as possible.
    \item[] Step 2. Compute the Choi matrix in the magic basis representation, $\{[\mathcal{J}(\mathcal{V}^{(\epsilon)}_s)]_{MB}\},$ using the relationship ~\eqref{eq:mb_choivec}.
    \item[] Step 3. Perform basis transformation to obtain $\{\mathcal{J}(\mathcal{V}_s^{\epsilon})\}$ in the computational basis, and then compute the matrix representation of shift unitaries $\{V_s^{\epsilon}\}$.
\end{enumerate}
In practice, we have employed numerical optimization for Step 1.

\section{White noise approximation under coherent errors}\label{sec:wn_coherent}
The white noise, or equivalently the global depolarizing noise on $d$-dimensional system acting as $\mathcal{E}_{wn}(\cdot) = (1-p)\cdot + p {\rm Id}/d$ under the rate $p$, is one of the most desirable error profile in the context of estimating expectation values of physical observables. 
The main reasons are two-fold: $\mathcal{E}_{wn}$ commutes with any gate operations, and it allows unbiased estimation with optimal sampling overhead.
All the expectation values of traceless observables shrink homogeneously by a factor of $1-p$ under a single application of $\mathcal{E}_{wn}$, and hence we can unbiasedly estimate the expectation value of an observable $O$ from a noisy estimation \black{as $\ev*{O} = (1-p)^{-1} \ev*{O}_{\rm noisy}$}. In the case when we have $L$ applications of $\mathcal{E}_{wn}$, we can simply replace \black{$(1-p)^{-1}$} with \black{$(1-p)^{-L}$ and mitigate error with minimal sampling overhead of $(1-p)^{-2L}$}.

The white noise has been argued to be realized in experiments, in particular when the unitary gates are chosen at random~\cite{arute2019quantum}.
In fact, one can prove that, even if the qubit connectivity is only linear and the circuit constitutes a brick-wall structure with random Haar unitary gates for two qubits, the effective noise converges to the white noise under unital noise~\cite{dalzell2024random, deshpande2022tight}.
%, while it is an open problem for non-unital noise~\cite{fefferman2023effect}.
\black{However, these arguments mainly consider incoherent error, and how coherent error can affect the white-noise approximation is not fully examined.
Given that the synthesis error is inevitable in the early FTQC regime, it is crucial to investigate the effect of coherent errors on the white-noise approximation.
}

\black{
Let us assume that one aims to simulate a $n$-qubit layered quantum circuit $\hat{\mathcal{U}}:=\mathcal{U}_{L+1} \circ\mathcal{U}_L \circ \cdots \circ \mathcal{U}_1$ which is exposed to noise channels as $\mathcal{U}=\mathcal{U}_{L+1} \circ\mathcal{E} \circ \mathcal{U}_L \circ \cdots \mathcal{E} \circ \mathcal{U}_1.$
We define the effective noise channel $\mathcal{E}'$ by re-expressing the noisy layered circuit as $\mathcal{U} = \mathcal{E}' \circ \hat{\mathcal{U}}$, 
When we expect that $\mathcal{E}'$ is close to the global white noise, we can mitigate errors and restore the ideal expectation value via some constant factor $R$, namely by rescaling the noisy one as $\mathrm{tr}[\rho O] =  R\mathrm{tr}[\mathcal{E}'(\rho) O]$.
Especially when each unitary layer $\mathcal{U}_l(\cdot) = U_l\cdot U_l^\dag$ is chosen randomly from $n$-qubit Haar random unitary, we can bound the bias between the ideal and the rescaled expectation value as
\begin{eqnarray}
    \label{eq_bias_general}
    \mathbb{E}_{U}\left[\big|R\mathrm{tr}[\mathcal{E}'(\rho)O] - \mathrm{tr}[\rho O]\big|\right]
    \leq \sqrt{1 - \qty(\frac{s^2}{u})^L },
\end{eqnarray}
by choosing the rescaling factor as $R = (s/u)^L$ (Theorem~S2 of \cite{tsubouchi2024symmetric}).
Here, $u$ and $s$ are noise-dependent constants called {\it unitarity} and {\it average noise strength}, respectively, which can be represented as
\begin{eqnarray}
    u &=& \frac{\sum_{ij} |\mathrm{tr}[E_iE_j^\dagger]|^2-1}{4^n-1} \\
    s &=& \frac{\sum_i |\mathrm{tr}[E_i]|^2 - 1}{4^n-1},
\end{eqnarray}
where $E_i$ is the Kraus operators of the noise $\mathcal{E}$.
}

\black{
For incoherent error such as the local depolarizing noise with uniform error rate $\mathcal{E} = \bigcirc_{i=1}^n \mathcal{E}^{(i)}_{\mathrm{dep}}$, where $\mathcal{E}^{(i)}_{\mathrm{dep}}(\cdot) = (1-\frac{3}{4}p)\cdot + \frac{p}{4}(X_i\cdot X_i + Y_i \cdot Y_i + Z_i \cdot Z_i)$ is the single-qubit local depolarizing channel on the $i$-th qubit, we obtain
\begin{equation}
    \mathbb{E}_{U}[|R\mathrm{tr}[\mathcal{E}'(\rho)O] - \mathrm{tr}[\rho O]|]
    \leq O(\sqrt{nL}p) \;\;\text{with}\;\; R=(1+O(p))^{nL}
\end{equation}
from Eq. \eqref{eq_bias_general}.
Meanwhile, for coherent error $\mathcal{E} = \bigcirc_{i=1}^n \mathcal{E}^{(i)}_{\mathrm{coh}}$ with $\|\mathcal{E}_{\rm coh}^{(i)} - {\rm Id}\|_{\diamond} = \epsilon_{\rm coh}$, we obtain
\begin{equation}
    \mathbb{E}_{U}[|R\mathrm{tr}[\mathcal{E}'(\rho)O] - \mathrm{tr}[\rho O]|]
    \leq O(\sqrt{nL}\epsilon_{\rm coh}) \;\;\text{with}\;\; R=(1+O(\epsilon_{\rm coh}^2))^{nL}.
\end{equation}
These results mean that, in terms of the sampling overhead $R^2$, coherent errors have quadratically suppressed sampling overhead compared to incoherent errors (see also Table~\ref{tbl_whitenoise}).
However, its effect on the bias is linear as in the case of incoherent errors.
Therefore, it is preferable to suppress its error rate through randomized compiling or error crafting, and improve the performance of the white-noise approximation.
}

\begin{table}[t]
  \caption{\black{Performance of the QEM through rescaling the noisy expectation value assuming the white-noise approximation for incoherent error with error rate $p$ and coherent error with diamond norm $\epsilon_{\rm coh}^2$. While coherent errors have quadratically suppressed sampling overhead, their effect on the bias is comparable with incoherent errors.}}
  \label{tbl_whitenoise}
  \centering
  \begin{tabular}{c||c|c}
      Error model & Bias & Sampling overhead  \\\hline\hline
      Incoherent error  & $O(\sqrt{nL}p)$ & $(1+O(p))^{2nL}$\\
      Coherent error   & $O(\sqrt{nL}\epsilon_{\rm coh})$ & $(1+O(\epsilon_{\rm coh}^2))^{2nL}$\\
  \end{tabular}
\end{table}

\black{We can observe a similar scaling numerically.}
Before discussing the effect of the coherent noise, let us briefly summarize the numerical argument for incoherent noise, \black{especially the local depolarizing noise $\mathcal{E} = \bigcirc_{i=1}^n \mathcal{E}^{(i)}_{\mathrm{dep}}$}, provided in Ref.~\cite{tsubouchiUniversalCostBound2023}.
We analyze the convergence to white noise by examining the Pauli transfer matrix (PTM) of the effective noise $\mathcal{E}'$, denoted hereafter as $\Gamma$.
We define $\Gamma$ by $\Gamma_{ij} = {\rm Tr}[P_i \mathcal{E}'(P_j)]/2^n$ for $P_i, P_j \in \{I, X, Y, Z\}^{\otimes n}$ for $n$-qubit system. Note that PTM is diagonal if $\mathcal{E}'$ is a Pauli channel, and furthermore there exists $ 0\leq w \leq 1$ such that $\Gamma_{ii}=w~(i\neq 0)$ if $\mathcal{E}'$ is the white noise.
In this regard, it is useful to compute the singular values of the PTM in order to investigate the closeness to the white noise.
We in particular compute  $\Gamma = V \Lambda W^\dagger$ with $\Lambda = {\rm diag}\{\lambda_i\}_{i=1}^{4^n}$, and define the damping factor $k_i$ as follows,
\begin{eqnarray}
    \lambda_i = (1-p)^{-k_i L}.
\end{eqnarray}
It can be proven that, when each unitary layer consists of unitary 2-design, the damping factor converges to $k_{\rm mean}= \frac{3n}{4} \frac{4^n}{4^n - 1}$, and numerical analysis supports that this also holds for weaker randomness, e.g., 2-local random circuits with linear connectivity, when the circuit depth is sufficiently deep.
In other words, the mean deviation $\mathbb{E}_i[|k_i - k_{\rm mean}|]$, as well as the deviation of extremals $|k_{\rm max(min)} - k_{\rm mean}|$, is suppressed as $O(1/\sqrt{L})$~\cite{tsubouchiUniversalCostBound2023}, which is consistent with the suppression of total variation distance in measurement probability distribution in the computational basis~\cite{dalzell2024random}.

Now we ask how the above behavior is affected by the presence of coherent errors.
Here, for the sake of simplicity of the argument, we sample $U_i$ uniformly from the entire unitary group, i.e., take an $n$-qubit  Haar random unitary, and consider a local noise as $\mathcal{E}_l = \bigcirc_{i=1}^n \mathcal{E}_{\rm coh}^{(i)} \circ \mathcal{E}_{\rm dep}^{(i)},$ where $\mathcal{E}_{\rm coh}^{(i)}$ is a random 1-qubit unitary such that $\|\mathcal{E}_{\rm coh}^{(i)} - {\rm Id}\|_{\diamond} = \epsilon_{\rm coh}.$
In practice, this can be realized by constraining the rotation angles under the axial decomposition of the unitary.

Interestingly, we find that the effect of coherent error on the magnitude of the white noise is quadratically suppressed, while that on the convergence rate is linear.
In Fig.~\ref{fig:whitenoise_convergence}(a)(b) we can see that, although the fluctuation in the damping factor is expanded at the shallow-depth regime, it readily converges at the rate of $O(1/\sqrt{L})$, which is in similar to the depolarizing-only case~\cite{tsubouchiUniversalCostBound2023}.
We observe that the deviation $\max_k\{|k - k_{\rm mean, dep}|\}$ is lifted approximately by a factor of $(\epsilon_{\rm coh}+\epsilon_{\rm dep})/\epsilon_{\rm dep}$ where $\epsilon_{\rm dep}:= \|\mathcal{E}_{\rm dep}^{(i)} - {\rm Id}\|_{\diamond}$, implying the slow down of the white noise convergence.
Meanwhile, as shown in Fig.~\ref{fig:whitenoise_convergence}(c), we notice that the converged effective channel is less affected, in a sense that $|k_{\rm mean, tot} - k_{\rm mean, dep}|\propto \epsilon^2.$
\black{Phenomenologically, we can understand that the effect of coherent noise is self-twirled by the circuit such that the diamond norm is suppressed quadratically as in the random compilation scheme~\cite{odake2024robust}.}
These results indicate that, although the presence of coherent error does not severely increase the sampling overhead to perform error countermeasures such as the rescaling technique, the adequacy of the approximation is degraded \black{and hence desirable to avoid them by error-crafted synthesis.}

\begin{figure}[h]
    \centering
    \includegraphics[width=0.98\linewidth]{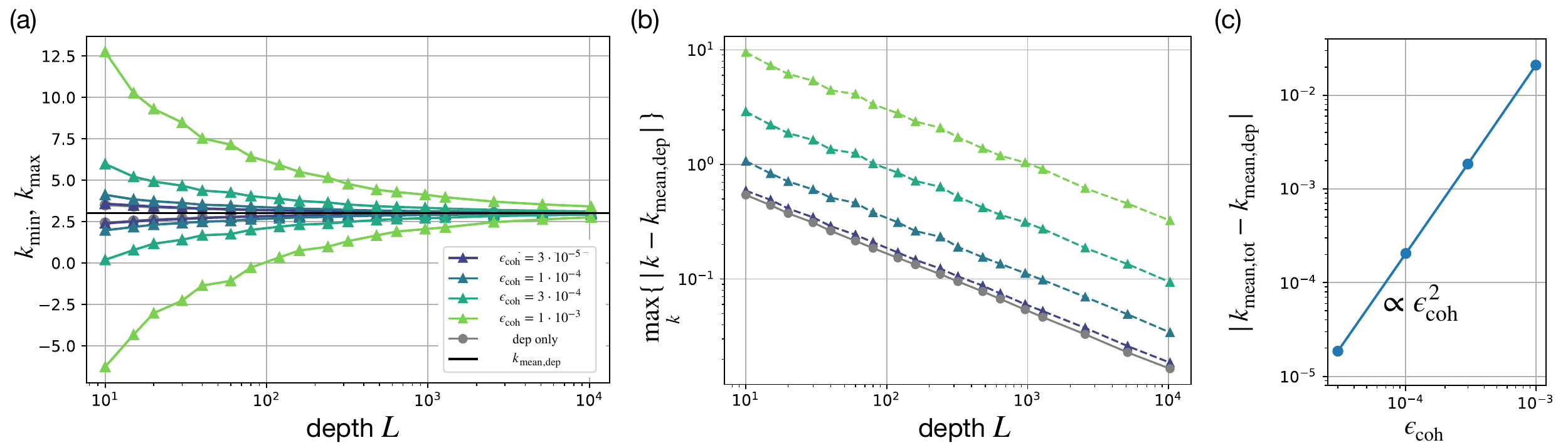}
    \caption{Convergence to white noise  in $n=4$ qubit quantum circuit. (a) Extremal values of damping factors $k_{\rm min}$ and $k_{\rm max}$ at various $L$. The unitary layer is drawn uniform-randomly from the $n$-qubit  unitary group. Each layer is subject to error $\bigcirc_{i=1}^N \mathcal{E}_{\rm coh} \circ \mathcal{E}_{\rm dep}$ with depolarizing error rate of $p=10^{-4}.$ (b) The largest deviation from $k_{\rm mean,dep},$ the value realized when only the depolarizing noise is present ($\epsilon_{\rm coh}=0$).
    (c) The quadratic suppression of damping factor $k$ with respect to the contribution from coherent noise. \black{Here, we compute $|k_{\rm mean, tot} - k_{\rm mean, dep}|$ at $L=10^4$.}
    }
    
    \label{fig:whitenoise_convergence}
\end{figure}

\section{Theoretical guarantees for error crafted synthesis} \label{sec:guarantee_crafting}
\subsection{Crafting Pauli channel (Proof of Theorem 1)}\label{subsec:guarantee_pauli_constraint}
In this section, we provide the formal statement and proof for Theorem~\ref{thm1_informal} in the main text.
First, we rewrite the Pauli-constrained minimization problem using the magic basis introduced in Sec.~\ref{sec:formalism_prob_synth} to obtain the following,
\begin{equation}
\label{eq:Paulierrorsimple}
 {\rm minimize}\ 1-\sum_j p_j\left(r^{(j)}_1\right)^2\ \ \ {\rm subject\ to}\ \ \sum_{j} p_j = 1, p_j \geq 0,\  \sum_{j} p_j r^{(j)}r^{(j)T}\ {\rm is\ diagonal},
\end{equation}
where $r^{(j)}=(r_1^{(j)},r_2^{(j)},r_3^{(j)},r_4^{(j)})\in\rr^4$ corresponds to the magic basis representation of the error channel for the $j$-th basis $\mathU_j\circ\hat{\mathU}^{-1}$. 
As we described in Sec.~\ref{sec:formalism_prob_synth}, this optimization can be solved by linear programming if the solution exists. Thus, the only nontrivial part for tailoring Pauli error is how to guarantee the existence of a solution.
As we already mentioned in the main text, we rigorously guarantee the existence of such a basis set:

\begin{thmS} (Detailed version of Theorem~\ref{thm1_informal}.)\label{thm2_new}
There exist a positive numbers $c$, $\epsilon_0$ and unit vectors $\{\hat{\vec{v}}^{(j)}\in\rr^3\}_{j=1}^7$ such that
for any given single-qubit unitary $\hat{\mathU}$ and for any $\epsilon\in(0,\epsilon_0]$, if we construct target unitaries $\{\hat{\mathU}_j\}_{j=1}^7$ satisfying
\begin{equation}
\frac{1}{2} \left[\mathJ(\hat{\mathU}_j\circ\hat{\mathU}^{-1})\right]_{MB}=\hat{e}^{(j)}\hat{e}^{(j)T}\ \  {\rm with} \ \  \hat{e}^{(j)}=
\begin{pmatrix}
 \sqrt{1-(c\epsilon)^2}\\
 c\epsilon\hat{\vec{v}}^{(j)}
\end{pmatrix}
\end{equation}
and obtain $\{\mathU_j\}_j$ via unitary synthesis such that $\frac{1}{2}\diamondnorm{\hat{\mathU}_j-\mathU_j}\leq\epsilon$, then a solution in Eq.~\eqref{eq:Paulierror} exists. Moreover, it holds that
\begin{equation}
\label{eq:proberrorPauli}
 ((c-1)\epsilon)^2 \leq \frac{1}{2}\diamondnorm{\hat{\mathU}-\sum_{j=1}^7 p_j \mathU_j}\leq((c+1)\epsilon)^2
\end{equation}
with the optimal probability distribution $\{p_j\}$ that induces a Pauli synthesis error.
\end{thmS}

\begin{proof}
 
By letting $r^{(j)}$ be the magic basis representation of $\mathU_j\circ\hat{\mathU}^{-1}$ and defining
\begin{equation}
\label{eq:defofAL}
 A_L=
\begin{pmatrix}
 r^{(1)}_1r^{(1)}_2&r^{(2)}_1r^{(2)}_2&\cdots&r^{(L)}_1r^{(L)}_2\\
 r^{(1)}_1r^{(1)}_3&r^{(2)}_1r^{(2)}_3&\cdots&r^{(L)}_1r^{(L)}_3\\
  r^{(1)}_1r^{(1)}_4&r^{(2)}_1r^{(2)}_4&\cdots&r^{(L)}_1r^{(L)}_4\\
  r^{(1)}_2r^{(1)}_3&r^{(2)}_2r^{(2)}_3&\cdots&r^{(L)}_2r^{(L)}_3\\
  r^{(1)}_2r^{(1)}_4&r^{(2)}_2r^{(2)}_4&\cdots&r^{(L)}_2r^{(L)}_4\\
  r^{(1)}_3 r^{(1)}_4&r^{(2)}_3r^{(2)}_4&\cdots&r^{(L)}_3r^{(L)}_4
\end{pmatrix}
=
\begin{pmatrix}
 m^{(1)}&m^{(2)}&\cdots&m^{(L)}
\end{pmatrix},
\end{equation}
we can describe the non-diagonal elements in $\sum_{j=1}^Lp_j r^{(j)}r^{(j)T}$ obtained by mixing $\{\mathU_j\}_{j=1}^L$ as a vector $A_Lp$.

Suppose that $A_6$ is invertible. Then, it is sufficient to show $(A_6^{-1}m^{(7)})_i\leq 0$ for any $i\in\{1,\cdots,6\}$. This is because $\mathbf{e}^{(i)T}A_6^{-1}m^{(k)}=\delta_{ik}$, where $\delta_{ik}$ is the Kronecker delta and $\mathbf{e}^{(i)}$ is the unit vector whose element is $0$ except the $i$-th element, 
\begin{eqnarray}
 \forall i\in\{1,\cdots,6\},\mathbf{e}^{(i)T}A_6^{-1}m^{(7)}\leq 0&\Leftrightarrow&\exists \tilde{p}:\{1,\cdots,6\}\rightarrow[0,\infty),m^{(7)}=-\sum_{j=1}^6\tilde{p}_j m^{(j)},
\end{eqnarray}
and we can verify $A_7p=0$ by setting $p_j=\frac{\tilde{p}_j}{1+\sum_{j'=1}^6\tilde{p}_{j'}}$ for $j\in\{1,\cdots,6\}$ and $p_7=\frac{1}{1+\sum_{j=1}^6\tilde{p}_j}$.

In the following, we show that $A_6$ is invertible and $\max_i(A_6^{-1}m^{(7)})_i\leq 0$ by using the continuity of the matrix inverse.
Define $D=diag(\frac{1}{c\epsilon\sqrt{1-(c\epsilon)^2}},\frac{1}{c\epsilon\sqrt{1-(c\epsilon)^2}},\frac{1}{c\epsilon\sqrt{1-(c\epsilon)^2}},\frac{1}{(c\epsilon)^2},\frac{1}{(c\epsilon)^2},\frac{1}{(c\epsilon)^2})$. By letting 
\begin{eqnarray}
\hat{\vec{v}}^{(1)}=
\begin{pmatrix}
-1\\0\\0 
\end{pmatrix},
\hat{\vec{v}}^{(2)}=
\begin{pmatrix}
0\\-1\\0 
\end{pmatrix},
\hat{\vec{v}}^{(3)}=
\begin{pmatrix}
0\\0\\1
\end{pmatrix},
\hat{\vec{v}}^{(4)}=
\frac{1}{\sqrt{2}}\begin{pmatrix}
1\\-1\\0 
\end{pmatrix},
\hat{\vec{v}}^{(5)}=
\frac{1}{\sqrt{2}}\begin{pmatrix}
-1\\0\\-1 
\end{pmatrix},
\hat{\vec{v}}^{(6)}=
\frac{1}{\sqrt{2}}\begin{pmatrix}
0\\1\\1 
\end{pmatrix},
\hat{\vec{v}}^{(7)}=
\frac{1}{\sqrt{3}}\begin{pmatrix}
1\\1\\-1 
\end{pmatrix},\nonumber\\
\end{eqnarray}
we can verify that
\begin{eqnarray}
 D\hat{A}_6=
 \begin{pmatrix}
 -1&0&0&\frac{1}{\sqrt{2}}&-\frac{1}{\sqrt{2}}&0\\
 0&-1&0&-\frac{1}{\sqrt{2}}&0&\frac{1}{\sqrt{2}}\\
 0&0&1&0&-\frac{1}{\sqrt{2}}&\frac{1}{\sqrt{2}}\\
 0&0&0&-\frac{1}{2}&0&0\\
  0&0&0&0&\frac{1}{2}&0\\
 0&0&0&0&0&\frac{1}{2}\\
\end{pmatrix},
 (D\hat{A}_6)^{-1}=
 \begin{pmatrix}
 -1&0&0&-\sqrt{2}&-\sqrt{2}&0\\
 0&-1&0&\sqrt{2}&0&\sqrt{2}\\
 0&0&1&0&\sqrt{2}&-\sqrt{2}\\
 0&0&0&-2&0&0\\
  0&0&0&0&2&0\\
 0&0&0&0&0&2\\
\end{pmatrix},
 \hat{A}_6^{-1}\hat{m}^{(7)}=
 -\frac{1}{3}\begin{pmatrix}
 \sqrt{3}\\
 \sqrt{3}\\
 \sqrt{3}\\
 2\\2\\2
\end{pmatrix},
\end{eqnarray}
where $\hat{A}_6$ is the matrix defined in the same way as Eq.~\eqref{eq:defofAL} for $\hat{\mathU}_j\circ\hat{\mathU}^{-1}$, and we use $ \hat{A}_6^{-1}\hat{m}^{(7)}= (D\hat{A}_6)^{-1}D\hat{m}^{(7)}$ in the last equality.
\black{Note that $\hat{A}_6$ also allows us to determine $\hat{m}^{(i)}~(i=1,...,6)$.}
If we can show that there exists an upper bound $u(c,\epsilon)$ of $\lpnorm{2}{Dm^{(j)}-D\hat{m}^{(j)}}$ satisfying
\begin{equation}
\label{eq:continuity_m}
\forall j,\lpnorm{2}{Dm^{(j)}-D\hat{m}^{(j)}}\leq u(c,\epsilon)\ \ \wedge\ \ 
\lim_{c\rightarrow\infty}\lim_{\epsilon\rightarrow0}u(c,\epsilon)=0,
\end{equation}
it implies that $DA_6$ is invertible ($\Leftrightarrow$ $A_6$ is invertible) and $\max_i(A_6^{-1}m^{(7)})_i=\max_i((DA_6)^{-1}Dm^{(7)})_i\leq 0$ for large enough $c$ and small enough $\epsilon$ since  $D\hat{A}_6$ is a constant invertible matrix (independent to $\epsilon$), $\max_i((D\hat{A}_6)^{-1}D\hat{m}^{(7)})_i=-\frac{1}{\sqrt{3}}$, the set of invertible matrices is open, and $f:A\rightarrow A^{-1}$ is continuous.

In the following, we assume that $c>1$ and $\epsilon_0<\frac{1}{2c}$ and consider the case when $\epsilon\in(0,\epsilon_0]$.
 Let $\sin\hat{\theta}=c\epsilon$, $\sin\theta_0=\epsilon$ and $r^{(j)}=(\cos \theta_j,\sin \theta_j\vec{v}^{(j)})^T$ with $\hat{\theta},\theta_0\in(0,\pi/6),\theta_j\in[0,\pi/2]$. 
Since the following argument holds for any $j$, we abbreviate a subscript or superscript indicating a label $j$.
 From the assumption of compilation accuracy, we obtain
 \begin{equation}
\cos\theta_0=\sqrt{1-\epsilon^2}\leq |e\cdot \hat{e}|=|\cos\theta\cos\hat{\theta}+\sin\theta\sin\hat{\theta}\vec{v}\cdot\hat{\vec{v}}|
=\cos\theta\cos\hat{\theta}+\sin\theta\sin\hat{\theta}\vec{v}\cdot\hat{\vec{v}},
\end{equation}
where we use $\theta<\frac{\pi}{3}$ (due to $|\theta-\hat{\theta}|<\theta_0$) in the last equality.
By using $\theta>0$ (due to $c>1$), this is equivalent to
\begin{equation}
  \label{eq:synthvec}
|\theta-\hat{\theta}|\leq\theta_0 \wedge\ \vec{v}\cdot\hat{\vec{v}}\geq\frac{\cos\theta_0-\cos\theta\cos\hat{\theta}}{\sin\theta\sin\hat{\theta}} (>0).
\end{equation}
On the other hand,
\begin{eqnarray}
 \lpnorm{2}{D(m-\hat{m})}&=&
 \sqrt{\lpnorm{2}{\hat{\vec{v}}-\frac{\cos\theta\sin\theta}{\cos\hat{\theta}\sin\hat{\theta}}\vec{v}}^2+\lpnorm{2}{
 \begin{pmatrix}
 \hat{v}_1\hat{v}_2\\\hat{v}_1\hat{v}_3\\\hat{v}_2\hat{v}_3
\end{pmatrix}
-
\left(\frac{\sin\theta}{\sin\hat{\theta}}\right)^2
 \begin{pmatrix}
 v_1v_2\\v_1v_3\\v_2v_3
\end{pmatrix}
 }^2}\\
&=&  \sqrt{\lpnorm{2}{\hat{\vec{v}}-\frac{\sin(2\theta)}{\sin(2\hat{\theta})}\vec{v}}^2+\frac{1}{2}\left(\lpnorm{2}{\hat{\vec{v}}\hat{\vec{v}}^T-\left(\frac{\sin\theta}{\sin\hat{\theta}}\right)^2\vec{v}\vec{v}^T}^2
-\sum_{i=1}^3\left(\hat{v}_i^2-\left(\frac{\sin\theta}{\sin\hat{\theta}}\right)^2v_i^2\right)^2\right)}\\
\label{eq:dm1}
&\leq&\sqrt{\lpnorm{2}{\hat{\vec{v}}-\frac{\sin(2\theta)}{\sin(2\hat{\theta})}\vec{v}}^2+\frac{1}{2}\left(\lpnorm{2}{\hat{\vec{v}}\hat{\vec{v}}^T-\left(\frac{\sin\theta}{\sin\hat{\theta}}\right)^2\vec{v}\vec{v}^T}^2
-\frac{1}{3}\left(1-\left(\frac{\sin\theta}{\sin\hat{\theta}}\right)^2\right)^2\right)},
\end{eqnarray}
where we use the Cauchy-Schwartz inequality $\left(\sum_{i=1}^3\left(\hat{v}_i^2-\frac{\sin^2\theta}{\sin^2\hat{\theta}}v_i^2\right)^2\right)\left(\sum_{i=1}^3\left(\frac{1}{\sqrt{3}}\right)^2\right)\geq\left(\sum_{i=1}^3\frac{1}{\sqrt{3}}\left(\hat{v}_i^2-\frac{\sin^2\theta}{\sin^2\hat{\theta}}v_i^2\right)\right)^2$ in the last inequality. 
%\if0 %%%%%%%%%%%% comment out since we have depolarizing first %%%%%%
We proceed with the calculation as follows:
\begin{eqnarray}
\label{eq:dm2}
 Eq.~\eqref{eq:dm1}&=&\sqrt{1+\left(\frac{\sin(2\theta)}{\sin(2\hat{\theta})}\right)^2-2\frac{\sin(2\theta)}{\sin(2\hat{\theta})}\vec{v}\cdot\hat{\vec{v}}
 +\frac{1}{2}\left(1+\left(\frac{\sin\theta}{\sin\hat{\theta}}\right)^4-2\left(\frac{\sin\theta}{\sin\hat{\theta}}\vec{v}\cdot\hat{\vec{v}}\right)^2\right)
 -\frac{1}{6}\left(1-\left(\frac{\sin\theta}{\sin\hat{\theta}}\right)^2\right)^2}.\nonumber\\
\end{eqnarray}
By using Eq.~\eqref{eq:synthvec}, we obtain
\begin{equation}
  \lpnorm{2}{D(m-\hat{m})}\leq\max_{|\theta-\hat{\theta}|\leq\theta_0}\sqrt{f(\hat{\theta},\theta,\theta_0)}=\sqrt{\max_{|t|\leq1}f(\hat{\theta},\theta_0t+\hat{\theta},\theta_0)},
\end{equation}
where $f(\hat{\theta},\theta,\theta_0)=1+\left(\frac{\sin(2\theta)}{\sin(2\hat{\theta})}\right)^2-2\frac{\sin(2\theta)}{\sin(2\hat{\theta})}\frac{\cos\theta_0-\cos\theta\cos\hat{\theta}}{\sin\theta\sin\hat{\theta}}
 +\frac{1}{2}\left(1+\left(\frac{\sin\theta}{\sin\hat{\theta}}\right)^4-2\left(\frac{\cos\theta_0-\cos\theta\cos\hat{\theta}}{\sin^2\hat{\theta}}\right)^2\right)
 -\frac{1}{6}\left(1-\left(\frac{\sin\theta}{\sin\hat{\theta}}\right)^2\right)^2$.
Define $g_c(\epsilon,t):=f(\sin^{-1}(c\epsilon),\sin^{-1}(\epsilon)t+\sin^{-1}(c\epsilon),\sin^{-1}(\epsilon))$ on $(\epsilon,t)\in(0,\epsilon_0]\times[-1,1]$.
To complete the proof of Eq.~\eqref{eq:continuity_m}, it is sufficient to show $\lim_{c\rightarrow\infty}\lim_{\epsilon\rightarrow0}\max_{|t|\leq1}g_c(\epsilon,t)=0$.
As shown in Section \ref{app:continuity}, $g_c(\epsilon,t)$ is uniformly continuous in $(\epsilon,t)\in(0,\epsilon_0]\times[-1,1]$. Thus, we can interchange $\lim_{\epsilon\rightarrow0}$ and $\max_{|t|\leq1}$ (see detail in Section \ref{app:commutativity_limmax}) and obtain
\begin{eqnarray}
\lim_{\epsilon\rightarrow0}\max_{|t|\leq1}g_c(\epsilon,t)
&=&\max_{|t|\leq1}\lim_{\epsilon\rightarrow0}g_c(\epsilon,t)\\
&=&\max_{|t|\leq1}\frac{1}{c^4}\left(\frac{1}{12}t^4+\frac{c}{3}t^3+\left(\frac{c^2}{3}+\frac{1}{2}\right)t^2+ct+2c^2-\frac{1}{4}\right)\\
&=&\frac{7}{3c^2}+\frac{4}{3c^3}+\frac{1}{3c^4},
\end{eqnarray}
where the last equality is obtained by observing that $t=1$ maximizes the polynomial via a straightforward calculation. 
%\fi %%%%%%%%%%%%%%%%%%%%%%%%%%
%By using the same argument as in Appendix~\ref{app:thm1}, we obtain
Finally, we obtain
\begin{equation}
\lpnorm{2}{Dm-D\hat{m}}\leq u(c,\epsilon)\ \ \wedge\ \ 
 \lim_{\epsilon\rightarrow0}u(c,\epsilon)=\sqrt{\frac{7}{3c^2}+\frac{4}{3c^3}+\frac{1}{3c^4}}.
\end{equation}
This completes the proof.
\end{proof}

\subsection{Crafting depolarizing channel} \label{subsec:guarantee_depol_constraint}
In similar to the Pauli constraint, the minimization problem to obtain the synthesis probability distribution under the depolarizing constraint can be written as
\begin{equation}
\label{eq:Depolerror}
{\rm minimize}\ 
\frac{1}{2}\diamondnorm{\hat{\mathU}-\sum_j p_j\mathU_j}\ \ \ {\rm subject\ to}\ \ 
\sum_{a} \chi_{aa}=1, \chi_{\rm XX} = \chi_{\rm YY} = \chi_{\rm ZZ} = q\geq 0, \chi_{a \neq b} = 0,
\end{equation}
which is rewritten under the magic basis representation as
\begin{equation}
 \label{eq:Depolerrorsimple}
 {\rm minimize}\ 1-\sum_j p_j\left(r^{(j)}_1\right)^2\ \ \ {\rm subject\ to}\ \ \sum_{j} p_j = 1, p_j \geq 0,\  \sum_{j}p_j r^{(j)}r^{(j)T}=(1-q)\mathbf{e}^{(1)}\mathbf{e}^{(1)T}+\frac{q}{4} {\rm Id},
\end{equation}
where $r^{(j)}\in\rr^4$ corresponds to the magic basis representation of $\mathU_j\circ\hat{\mathU}^{-1}$ given in Eq.~\eqref{eq:MBrep}, and $\mathbf{e}^{(i)}$ is the unit vector whose element is $0$ except the $i$-th element.

We can prove the feasibility guarantee as in the following theorem:
\begin{thmS}\label{thm1_new}
There exist positive numbers $c$, $\epsilon_0$ and unit vectors $\{\hat{\vec{v}}^{(j)}\in\rr^3\}_{j=1}^9$ such that
for any given single-qubit unitary $\hat{\mathU}$ and for any $\epsilon\in(0,\epsilon_0]$, if we prepare target unitaries $\{\hat{\mathU}_j\}_{j=1}^9$ satisfying
\begin{equation}
\frac{1}{2} \left[\mathJ(\hat{\mathU}_j\circ\hat{\mathU}^{-1})\right]_{MB}=\hat{e}^{(j)}\hat{e}^{(j)T}\ \  {\rm with} \ \  \hat{e}^{(j)}=
\begin{pmatrix}
 \sqrt{1-(c\epsilon)^2}\\
 c\epsilon\hat{\vec{v}}^{(j)}
\end{pmatrix}
\end{equation}
and obtain $\{\mathU_j\}_j$ via unitary synthesis such that $\frac{1}{2}\diamondnorm{\hat{\mathU}_j-\mathU_j}\leq\epsilon$, then a solution in Eq.~\eqref{eq:Depolerror} exists. Moreover, it holds that
\begin{equation}
\label{eq:proberrorDepo}
%((c-1)\epsilon)^2 \leq
 ((c-1)\epsilon)^2 \leq  \frac{1}{2}\diamondnorm{\hat{\mathU}-\sum_{j=1}^9 p_j\mathU_j}\leq((c+1)\epsilon)^2
\end{equation}
with the optimal probability distribution $\{p_j\}$ that induces a depolarizing channel as the effective synthesis error.
\end{thmS}

\begin{proof}
     By letting $r^{(j)}$ be the magic basis representation of $\mathU_j\circ\hat{\mathU}^{-1}$ and defining
\begin{equation}
 A_L=
\begin{pmatrix}
 r^{(1)}_1r^{(1)}_2&r^{(2)}_1r^{(2)}_2&\cdots&r^{(L)}_1r^{(L)}_2\\
 r^{(1)}_1r^{(1)}_3&r^{(2)}_1r^{(2)}_3&\cdots&r^{(L)}_1r^{(L)}_3\\
  r^{(1)}_1r^{(1)}_4&r^{(2)}_1r^{(2)}_4&\cdots&r^{(L)}_1r^{(L)}_4\\
  r^{(1)}_2r^{(1)}_3&r^{(2)}_2r^{(2)}_3&\cdots&r^{(L)}_2r^{(L)}_3\\
  r^{(1)}_2r^{(1)}_4&r^{(2)}_2r^{(2)}_4&\cdots&r^{(L)}_2r^{(L)}_4\\
  r^{(1)}_3r^{(1)}_4&r^{(2)}_3r^{(2)}_4&\cdots&r^{(L)}_3r^{(L)}_4\\
  (r^{(1)}_2)^2-(r^{(1)}_3)^2&(r^{(2)}_2)^2-(r^{(2)}_3)^2&\cdots&(r^{(L)}_2)^2-(r^{(L)}_3)^2\\
  (r^{(1)}_2)^2-(r^{(1)}_4)^2&(r^{(2)}_2)^2-(r^{(2)}_4)^2&\cdots&(r^{(L)}_2)^2-(r^{(L)}_4)^2
\end{pmatrix}
=
\begin{pmatrix}
 m^{(1)}&m^{(2)}&\cdots&m^{(L)}
\end{pmatrix},
\end{equation}
we can describe the restrictions about the depolarizing error obtained by mixing $\{\mathU_j\}_{j=1}^L$ as $A_Lp=0$.

Define $D=diag(\frac{1}{c\epsilon\sqrt{1-(c\epsilon)^2}},\frac{1}{c\epsilon\sqrt{1-(c\epsilon)^2}},\frac{1}{c\epsilon\sqrt{1-(c\epsilon)^2}},\frac{1}{(c\epsilon)^2},\frac{1}{(c\epsilon)^2},\frac{1}{(c\epsilon)^2},\frac{1}{\sqrt{8}(c\epsilon)^2},\frac{1}{\sqrt{8}(c\epsilon)^2})$.
Suppose that $DA_8$ is invertible. Then, it is sufficient to show $((DA_8)^{-1}(Dm^{(9)}))_i\leq 0$ for any $i\in\{1,\cdots,8\}$ and $\{r^{(j)}\}_{j=1}^9$. 
We can find $\{\hat{v}^{(x)}\}_{x=1}^9$ satisfying the conditions such that $D\hat{A}_8$ is a constant invertible matrix and $\max_i((D\hat{A}_8)^{-1}(D\hat{m}^{(9)}))_i$ is a constant negative real. For example, by setting
\begin{eqnarray}
    \hat{\vec{v}}^{(1)}=
\begin{pmatrix}
1\\0\\0 
\end{pmatrix},
\hat{\vec{v}}^{(2)}=
\begin{pmatrix}
-1\\0\\0 
\end{pmatrix},
\hat{\vec{v}}^{(3)}=
\begin{pmatrix}
0\\-1\\0
\end{pmatrix},
\hat{\vec{v}}^{(4)}=
\begin{pmatrix}
0\\0\\1
\end{pmatrix},
\hat{\vec{v}}^{(5)}=
\begin{pmatrix}
0\\0\\-1
\end{pmatrix},\nonumber\\
\hat{\vec{v}}^{(6)}=
\frac{1}{\sqrt{2}}\begin{pmatrix}
-1\\0\\-1 
\end{pmatrix},
\hat{\vec{v}}^{(7)}=
\frac{1}{\sqrt{2}}\begin{pmatrix}
0\\1\\1 
\end{pmatrix},
\hat{\vec{v}}^{(8)}=
\frac{1}{\sqrt{2}}\begin{pmatrix}
-1\\1\\0 
\end{pmatrix},
\hat{\vec{v}}^{(7)}=
\frac{1}{\sqrt{3}}\begin{pmatrix}
1\\1\\-1 
\end{pmatrix},\nonumber\\
\end{eqnarray}
we can verify the conditions are satisfied. This is because
\begin{eqnarray}
 D\hat{A}_8=
 \begin{pmatrix}
 1&-1&0&0&0&-\frac{1}{\sqrt{2}}&0&-\frac{1}{\sqrt{2}}\\
 0&0&-1&0&0&0&\frac{1}{\sqrt{2}}&\frac{1}{\sqrt{2}}\\
 0&0&0&1&-1&-\frac{1}{\sqrt{2}}&\frac{1}{\sqrt{2}}&0\\
 0&0&0&0&0&0&0&-\frac{1}{2}\\
  0&0&0&0&0&\frac{1}{2}&0&0\\
 0&0&0&0&0&0&\frac{1}{2}&0\\
 \frac{1}{2\sqrt{2}}&\frac{1}{2\sqrt{2}}&-\frac{1}{2\sqrt{2}}&0&0&\frac{1}{4\sqrt{2}}&-\frac{1}{4\sqrt{2}}&0\\
 \frac{1}{2\sqrt{2}}&\frac{1}{2\sqrt{2}}&0&-\frac{1}{2\sqrt{2}}&-\frac{1}{2\sqrt{2}}&0&-\frac{1}{4\sqrt{2}}&\frac{1}{4\sqrt{2}}
\end{pmatrix},\nonumber\\
 (D\hat{A}_8)^{-1}=
 \begin{pmatrix}
 \frac{1}{2}&-\frac{1}{2}&0&-\sqrt{2}&-\frac{1}{2}+\frac{1}{\sqrt{2}}&\frac{1}{2}+\frac{1}{\sqrt{2}}&\sqrt{2}&0\\
 -\frac{1}{2}&-\frac{1}{2}&0&0&-\frac{1}{2}-\frac{1}{\sqrt{2}}&\frac{1}{2}+\frac{1}{\sqrt{2}}&\sqrt{2}&0\\
 0&-1&0&-\sqrt{2}&0&\sqrt{2}&0&0\\
 0&-\frac{1}{2}&\frac{1}{2}&-\frac{1}{2}-\frac{1}{\sqrt{2}}&-\frac{1}{2}+\frac{1}{\sqrt{2}}&0&\sqrt{2}&-\sqrt{2}\\
 0&-\frac{1}{2}&-\frac{1}{2}&-\frac{1}{2}-\frac{1}{\sqrt{2}}&-\frac{1}{2}-\frac{1}{\sqrt{2}}&\sqrt{2}&\sqrt{2}&-\sqrt{2}\\
 0&0&0&0&2&0&0&0\\
 0&0&0&0&0&2&0&0\\ 
 0&0&0&-2&0&0&0&0
\end{pmatrix},
 \hat{A}_8^{-1}\hat{m}^{(9)}=
 -\frac{1}{3}\begin{pmatrix}
 2\sqrt{2}\\
 \sqrt{3}\\
 2\sqrt{2}\\
 \sqrt{3}\\
 \sqrt{2}+\sqrt{3}\\
  \sqrt{2}\\
 2\\2\\2
\end{pmatrix},
\end{eqnarray}
where we use $ \hat{A}_8^{-1}\hat{m}^{(9)}= (D\hat{A}_8)^{-1}D\hat{m}^{(9)}$ in the last equality.

Thus, it is sufficient to show there exists an upper bound $u(c,\epsilon)$ of $\lpnorm{2}{Dm^{(j)}-D\hat{m}^{(j)}}$ satisfying
\begin{equation}
\label{eq:continuity_dm}
\lpnorm{2}{Dm^{(j)}-D\hat{m}^{(j)}}\leq u(c,\epsilon)\ \ \wedge\ \ 
\lim_{c\rightarrow\infty}\lim_{\epsilon\rightarrow0}u(c,\epsilon)=0,
\end{equation}
In the following, we assume that $c>1$ and $\epsilon_0<\frac{1}{2c}$ and consider the case when $\epsilon\in(0,\epsilon_0]$.
 Let $\sin\hat{\theta}=c\epsilon$, $\sin\theta_0=\epsilon$ and $r^{(j)}=(\cos \theta_j,\sin \theta_j\vec{v}^{(j)})^T$ with $\hat{\theta},\theta_0\in(0,\pi/6),\theta_x\in[0,\pi/2]$. 
Since the following argument holds for any $j$, we abbreviate a subscript or superscript indicating a label $j$.
 From the assumption of compilation accuracy, we obtain
\begin{equation}
  \label{eq:synthvec2}
|\theta-\hat{\theta}|\leq\theta_0 \wedge\ \vec{v}\cdot\hat{\vec{v}}\geq\frac{\cos\theta_0-\cos\theta\cos\hat{\theta}}{\sin\theta\sin\hat{\theta}} (>0).
\end{equation}

On the other hand,
\begin{eqnarray}
 \lpnorm{2}{D(m-\hat{m})}^2&=&
 \lpnorm{2}{\hat{\vec{v}}-\frac{\cos\theta\sin\theta}{\cos\hat{\theta}\sin\hat{\theta}}\vec{v}}^2
 +\lpnorm{2}{
 \begin{pmatrix}
 \hat{v}_1\hat{v}_2\\\hat{v}_1\hat{v}_3\\\hat{v}_2\hat{v}_3
\end{pmatrix}
-
\left(\frac{\sin\theta}{\sin\hat{\theta}}\right)^2
 \begin{pmatrix}
 v_1v_2\\v_1v_3\\v_2v_3
\end{pmatrix}
 }^2
 +\frac{1}{8}\lpnorm{2}{
  \begin{pmatrix}
 \hat{v}_1^2-\hat{v}_2^2\\\hat{v}_1^2-\hat{v}_3^2
\end{pmatrix}
-\left(\frac{\sin\theta}{\sin\hat{\theta}}\right)^2
  \begin{pmatrix}
 v_1^2-v_2^2\\v_1^2-v_3^2
\end{pmatrix}
 }^2\nonumber\\\\
 &=& \lpnorm{2}{\hat{\vec{v}}-\frac{\sin(2\theta)}{\sin(2\hat{\theta})}\vec{v}}^2+\frac{1}{2}\left(\lpnorm{2}{\hat{\vec{v}}\hat{\vec{v}}^T-\left(\frac{\sin\theta}{\sin\hat{\theta}}\right)^2\vec{v}\vec{v}^T}^2
-\sum_{i=1}^3r_i^2\right)
+\frac{1}{8}\left((r_1-r_2)^2+(r_1-r_3)^2\right)\\
&\leq& \lpnorm{2}{\hat{\vec{v}}-\frac{\sin(2\theta)}{\sin(2\hat{\theta})}\vec{v}}^2+\frac{1}{2}\lpnorm{2}{\hat{\vec{v}}\hat{\vec{v}}^T-\left(\frac{\sin\theta}{\sin\hat{\theta}}\right)^2\vec{v}\vec{v}^T}^2
+\frac{(r_1-r_2)^2+(r_1-r_3)^2+(r_2-r_3)^2-4\sum_{i=1}^3r_i^2}{8}\\
&=& \lpnorm{2}{\hat{\vec{v}}-\frac{\sin(2\theta)}{\sin(2\hat{\theta})}\vec{v}}^2+\frac{1}{2}\lpnorm{2}{\hat{\vec{v}}\hat{\vec{v}}^T-\left(\frac{\sin\theta}{\sin\hat{\theta}}\right)^2\vec{v}\vec{v}^T}^2
-\frac{(r_1+r_2)^2+(r_1+r_3)^2+(r_2+r_3)^2}{8}\\
&\leq& \lpnorm{2}{\hat{\vec{v}}-\frac{\sin(2\theta)}{\sin(2\hat{\theta})}\vec{v}}^2+\frac{1}{2}\lpnorm{2}{\hat{\vec{v}}\hat{\vec{v}}^T-\left(\frac{\sin\theta}{\sin\hat{\theta}}\right)^2\vec{v}\vec{v}^T}^2
-\frac{1}{6}\left(1-\left(\frac{\sin\theta}{\sin\hat{\theta}}\right)^2\right)^2,
\end{eqnarray}
where we let $r_i=\hat{v}_i^2-\left(\frac{\sin\theta}{\sin\hat{\theta}}\right)^2v_i^2$ and 
we use the Cauchy-Schwartz inequality 
\begin{eqnarray}
\left((r_1+r_2)^2+(r_1+r_3)^2+(r_2+r_3)^2\right)\left(\sum_{i=1}^3\left(\frac{1}{\sqrt{3}}\right)^2\right)\geq
\left(\frac{r_1+r_2}{\sqrt{3}}+\frac{r_1+r_3}{\sqrt{3}}+\frac{r_2+r_3}{\sqrt{3}}\right)^2=\frac{4}{3}\left(1-\left(\frac{\sin\theta}{\sin\hat{\theta}}\right)^2\right)^2.    
\end{eqnarray} 
By using the same argument as in Sec.~\ref{subsec:guarantee_pauli_constraint}, we obtain
\begin{equation}
\lpnorm{2}{Dm-D\hat{m}}\leq u(c,\epsilon)\ \ \wedge\ \ 
 \lim_{\epsilon\rightarrow0}u(c,\epsilon)=\sqrt{\frac{7}{3c^2}+\frac{4}{3c^3}+\frac{1}{3c^4}}.
\end{equation}
%%%%%%%%%%%%%%%%%%%%%%%%%%%%%%%%%%%%%%%%%%%%%%%
This completes the proof of Eq.~\eqref{eq:continuity_dm}.
\end{proof}

\section{Technical lemmas}\label{app:technical-lemma}

\subsection{Uniform continuity of $g_c(\epsilon,t)$} \label{app:continuity}
Since the summation $f+g$, the product $fg$ and the composition $f\circ g$ of two bounded uniformly continuous functions $f$ and $g$ are bounded and uniformly continuous, we can show the uniform continuity of $g_c(\epsilon,t)$ by observing that its each term is bounded and uniformly continuous on $(\epsilon,t)\in(0,\epsilon_0]\times[-1,1]$ with $\epsilon_0<\frac{1}{2c}$ and $c>1$ as follows:
\begin{enumerate}
\item $\theta=\sin^{-1}(\epsilon)t+\sin^{-1}(c\epsilon)$, $\hat{\theta}=\sin^{-1}(c\epsilon)$ and $\theta_0=\sin^{-1}(\epsilon)$ are bounded and uniformly continuous, whose range lies in $(0,\frac{\pi}{3})$.

\item $\frac{\theta_0}{\hat{\theta}}=\frac{1}{c}\frac{\sin^{-1}(\epsilon)}{\epsilon}\frac{c\epsilon}{\sin^{-1}(c\epsilon)}$ and 
$\frac{\theta}{\hat{\theta}}=\frac{\theta_0}{\hat{\theta}}t+1$ are bounded and uniformly continuous since $\frac{\sin^{-1}(r)}{r}$ and $\frac{r}{\sin^{-1}(r)}$ are bounded and uniformly continuous on  $r\in(0,1]$.
Since both ranges of $\frac{\theta_0}{\hat{\theta}}$ and $\frac{\theta}{\hat{\theta}}$ are included in a closed interval $[a,b]\subset(0,\infty)$ and $f(r)=\frac{1}{r}$ is bounded and uniformly continuous on $r\in[a,b]$, $\frac{\hat{\theta}}{\theta_0}$ and $\frac{\hat{\theta}}{\theta}$ are also bounded and uniformly continuous.
By using the similar argument, we can verify that $\frac{\theta}{\theta_0}=t+\frac{\hat{\theta}}{\theta_0}$ and $\frac{\theta_0}{\theta}$ are bounded and uniformly continuous.

 \item $\frac{\sin\theta}{\sin\hat{\theta}}=\frac{\sin\theta}{\theta}\frac{\hat{\theta}}{\sin\hat{\theta}}\frac{\theta}{\hat{\theta}}$ and $\frac{\sin(2\theta)}{\sin(2\hat{\theta})}=\frac{\sin(2\theta)}{2\theta}\frac{2\hat{\theta}}{\sin(2\hat{\theta})}\frac{\theta}{\hat{\theta}}$ are bounded and uniformly continuous since $\frac{\sin \phi}{\phi}$ and $\frac{\phi}{\sin \phi}$ are bounded and uniformly continuos on $\phi\in(0,\frac{2\pi}{3})$.
 
 \item $\frac{\cos\theta_0-\cos\theta\cos\hat{\theta}}{\sin\theta\sin\hat{\theta}}=\frac{\theta}{\sin\theta}\frac{\hat{\theta}}{\sin\hat{\theta}}\left(\frac{\theta_0}{\theta}\frac{\theta_0}{\hat{\theta}}\frac{\cos\theta_0-1}{\theta_0^2}
 -\frac{\hat{\theta}}{\theta}\frac{\cos\hat{\theta}-1}{\hat{\theta}^2}
 +\frac{\theta}{\hat{\theta}}\frac{(1-\cos\theta)}{\theta^2}\cos\hat{\theta}\right)$ is bounded and uniformly continuous since $\frac{\cos\phi-1}{\phi^2}$ is bounded and uniformly continuous on $\phi\in(0,\frac{\pi}{3})$.

\end{enumerate}

\subsection{Commutativity of $\lim$ and $\max$ for uniformly continuous functions} \label{app:commutativity_limmax}
Let $E$ and $T$ be an arbitrary subset and a compact subset of $\rr$, respectively. Let $\epsilon_0$ be a limit point of $E$ and $g:E\times T\rightarrow\rr$ be uniformly continuous such that $\lim_{\epsilon\rightarrow\epsilon_0}g(\epsilon,t)$ is finite for all $t\in T$. Then, we can show that
\begin{equation}
\label{eq:limmax}
 \lim_{\epsilon\rightarrow\epsilon_0}\max_{t\in T}g(\epsilon,t)=\max_{t\in T}\lim_{\epsilon\rightarrow\epsilon_0}g(\epsilon,t).
\end{equation}
Note that the 'uniform' continuity is requisite, e.g., $\lim_{\epsilon\rightarrow0}\max_{t\in T}g(\epsilon,t)=1$ and $\max_{t\in T}\lim_{\epsilon\rightarrow0}g(\epsilon,t)=0$ for $T=[0,1]$, $E=(0,1)$, and $g$ is a continuous function such that $\supp{g}$ is included in $\{(\epsilon,t):t<\epsilon<4t\}$ and $g(\epsilon,t)=1$ if $2t\leq\epsilon\leq3t$.
\begin{proof}
 First, we can verify that $f(t):=\lim_{\epsilon\rightarrow\epsilon_0}g(\epsilon,t)$ and $h(\epsilon):=\max_{t\in T}g(\epsilon,t)$ are uniformly continuous as follows.
 Since $g$ is uniformly continuous,
 for any $\Delta>0$, there exists $\delta>0$ such that for any $t,t'\in T$ and $\epsilon,\epsilon'>0$, if $
 \lpnorm{2}{
 \begin{pmatrix}
 \epsilon\\ t
\end{pmatrix}
-
\begin{pmatrix}
 \epsilon'\\ t'
\end{pmatrix}
}<\delta
$, $|g(\epsilon,t)-g(\epsilon',t')|<\Delta$.
This implies that if $|t-t'|<\delta$, $|f(t)-f(t')|=\left|\lim_{\epsilon\rightarrow\epsilon_0}\left(g(\epsilon,t)-g(\epsilon,t')\right)\right|
 \leq\Delta$. It also implies that if $|\epsilon-\epsilon'|<\delta$, $|h(\epsilon)-h(\epsilon')|\leq\max_{t\in T}|g(\epsilon,t)-g(\epsilon',t)|<\Delta$.
 Since $f$ is continuous and $h$ is uniformly continuous, both sides in Eq.~\eqref{eq:limmax} are well-defined.
 
Since $(L.H.S.)\geq(R.H.S.)$ is trivial, we show the converse. For any $\Delta>0$, there exists $\delta>0$ such that $|f(t)-g(\epsilon,t)|<\Delta$ if $|\epsilon-\epsilon_0|\leq\delta$ for any $t\in T$ and $\epsilon$. Thus, we obtain
\begin{equation}
  \lim_{\epsilon\rightarrow\epsilon_0}\max_{t\in T}g(\epsilon,t)\leq  \lim_{\epsilon\rightarrow\epsilon_0}\max_{t\in T}(f(t)+\Delta)=\max_{t\in T}\lim_{\epsilon\rightarrow\epsilon_0}g(\epsilon,t)+\Delta.
\end{equation}
Since this holds for any $\Delta>0$, this completes the proof.
\end{proof}

\end{document}